\documentclass[authoryear,review, 12pt]{elsarticle}

\newcommand{\maxwidth}{\textwidth}

\usepackage{alltt}
\usepackage[T1]{fontenc}
\usepackage[latin9]{inputenc}
\usepackage{geometry}
\usepackage{bm}
\usepackage{amsthm}
\usepackage{amsmath}
\usepackage{amssymb}
\usepackage{undertilde}
\usepackage{graphicx}
\usepackage{setspace}
\usepackage{esint}
\usepackage{booktabs}
\usepackage{color}
\usepackage{multirow}
\usepackage{natbib}
\usepackage{ifthen}

\mathchardef\mhyphen="2D 

\newtheorem{thm}{Theorem}

\journal{arXiv}

\newboolean{thesis}
\setboolean{thesis}{false}

\begin{document}

\begin{frontmatter}

\title{Local Adaptive Grouped Regularization and its Oracle Properties for Varying Coefficient Regression}

\author[uw]{Wesley Brooks}
\ead{wrbrooks@uwalumni.com}

\author[uw]{Jun Zhu}
\ead{jzhu@stat.wisc.edu}

\author[zlu]{Zudi Lu}
\ead{Z.Lu@soton.ac.uk}

\address[uw]{Department of Statistics, University of Wisconsin, Madison, WI 53706}
\address[zlu]{School of Mathematical Sciences, University of Southampton, Highfield, Southampton UK}

\begin{abstract}
Varying coefficient regression is a flexible technique for modeling data where the coefficients are functions of some effect-modifying parameter, often time or location in a certain domain. While there are a number of methods for variable selection in a varying coefficient regression model, the existing methods are mostly for global selection, which includes or excludes each covariate over the entire domain. Presented here is a new local adaptive grouped regularization (LAGR) method for local variable selection in spatially varying coefficient linear and generalized linear regression. LAGR selects the covariates that are associated with the response at any point in space, and simultaneously estimates the coefficients of those covariates by tailoring the adaptive group Lasso toward a local regression model with locally linear coefficient estimates. Oracle properties of the proposed method are established under local linear regression and local generalized linear regression. The finite sample properties of LAGR are assessed in a simulation study and for illustration, the Boston housing price data set is analyzed.
\end{abstract}

\begin{keyword}
adaptive Lasso, local generalized linear regression, local polynomial regression, nonparametric, regularization method, varying coefficient model
\end{keyword}

\end{frontmatter}

\section{Introduction}

Whereas the coefficients in traditional linear regression are scalar
constants, the coefficients in a varying coefficient regression (VCR)
model are functions - often \emph{smooth} functions - of some effect-modifying
parameter \citep{Cleveland-Grosse-1991,Hastie-Tibshirani-1993}. Here
we treat the case of a VCR model on a spatial domain where the spatial
location is a two-dimensional effect-modifying parameter. Current
practice for VCR models relies on global model selection to decide
which variables should be included in the model, meaning that covariates
are selected for inclusion or exclusion over the entire domain.
Various methods have been developed by using, for example, P-splines
\citep{Antoniadis:2012a}, basis expansion \citep{Wang-2008a}, and
local regression \citep{Wang-Xia-2009}. Since the coefficients vary
in a VCR model, in principle there is no reason that the best model
must use the same set of covariates everywhere on the domain - that
is, some of the coefficients may be zero in part of the domain. New
methodology is developed here for guiding the decision of which covariates
belong in the VCR model at any location, termed local variable selection,
as the literature on how to do so is currently scarce.
Such new methodology provides for more flexible variable selection in regression models
with coefficients that vary in space.

Specifically, local adaptive grouped regularization (LAGR) is developed
here as a novel method of local variable selection at a given location
in the domain of a VCR model. The method of LAGR applies to VCR models
where the coefficients are estimated using locally linear kernel smoothing.
Kernel smoothing for nonparametric regression is described in detail
in \citet*{Fan-Gijbels-1996}. The extension to estimating VCR models
is made by \citet{Fan-Zhang-1999} for a VCR model with a univariate
effect-modifying parameter, and by \citet{Sun-Yan-Zhang-Lu-2014}
for a two-dimensional effect-modifying parameter in a spatial VCR
with autocorrelation. These methods mitigate the boundary effect by
estimating the coefficients as local polynomials of odd degree (usually
locally linear) \citep{Hastie:1993b}. However, none of these authors addressed local variable selection.
In this work, we focus on local variable selection with a
two-dimensional effect-modifying parameter and discuss the effect
of dimensionality on the results.

For standard linear regression models, the least absolute shrinkage
and selection operator (Lasso) is a regularization method that
simultaneously selects covariates for the regression model and shrinks
the coefficient estimates toward zero \citep{Tibshirani-1996}. However,
the Lasso can be inconsistent for variable selection and inefficient
for coefficient estimation \citep{Zou-2006}. The adaptive Lasso (AL)
is a refinement of the Lasso that produces consistent estimates of
the coefficients and has been shown to have appealing properties for
variable selection, which under suitable conditions include the ``oracle''
property of asymptotically including exactly the correct set of covariates
and estimating their coefficients as well as if the correct covariates
were known in advance \citep{Zou-2006}. For data where the observed
covariates fall into mutually exclusive groups that are known in advance,
the adaptive group Lasso has similar oracle properties to the adaptive
Lasso but selects groups rather than individual covariates
\citep{Yuan-Lin-2006,Wang-Leng-2008}. An innovation here is to develop an adaptive group Lasso for
local variable selection and coefficient estimation in a locally linear
regression model, where each group consists of a single covariate
and its interactions with the effect-modifying parameter. Further, we consider both
varying coefficient linear regression for Gaussian response and varying coefficient
generalized linear regression for responses that are not necessarily
Gaussian. We show that the proposed LAGR method possesses the oracle properties of asymptotically
selecting exactly the correct local covariates and estimating their
local coefficients as accurately as would be possible if the identity
of the nonzero coefficients for the local model were known in advance.

The remainder of this paper is organized as follows. The kernel-based
estimation of a VCR model is described in Section \ref{sec:vcr}.
The proposed LAGR technique for varying coefficient linear regression
and its oracle properties are presented in Section \ref{sec:lagr-gaussian}.
In Section \ref{sec:simulations}, the finite-sample properties of
LAGR are evaluated in a simulation study, and in Section \ref{sec:example}
LAGR is applied to the Boston housing price dataset. In Section \ref{sec:lagr-gllm},
LAGR is extended to varying coefficient generalized linear regression
and the oracle properties for this setting are established, followed
by conclusions and discussion in Section 7. Technical proofs are given in the appendices.

\section{Varying Coefficient Regression\label{sec:vcr}}

\subsection{Varying Coefficient Model}

Consider $n$ observation locations $\bm{s}_{i}=(s_{i,1},\; s_{i,2})^{T}$
for $i=1,\dots,n$, which are distributed in a domain $\mathcal{D}\subset\mathbb{R}^{2}$
according to a density $f$. For $i=1,\dots,n$, let $Y_{i}=Y(\bm{s}_{i})$
and $\bm{X}_{i}=\bm{X}(\bm{s}_{i})$ denote, respectively, the univariate
response and the $(p+1)$--vector of covariates measured at location
$\bm{s}_{i}$. At location $\bm{s}_{i}$, assume that the outcome
is related to the covariates by a linear regression where the coefficients
$\bm{\beta}(\bm{s}_{i})$ are functions in the two dimensions and
$\varepsilon_{i}$ is random error at location $\bm{s}_{i}$. That
is, 
\begin{align}
Y_{i}=\bm{X}_{i}^{T}\bm{\beta}(\bm{s}_{i})+\varepsilon_{i}.\label{eq:lm(s)}
\end{align}

Further assume that the error term $\varepsilon_{i}$ is normally
distributed with zero mean and variance $\sigma^{2}$, and that $\varepsilon_{i}$,
$i=1,\dots,n$ are independent. That is, for $\bm{\varepsilon}=\left(\varepsilon_{1},\dots,\varepsilon_{n}\right)^{T}$,
$\bm{\varepsilon}\sim N\left(\bm{0},\sigma^{2}\bm{I}_{n}\right)$
where $\bm{I}_{n}$ denotes the identity matrix. 

In the context of nonparametric regression, the boundary-effect bias
can be reduced by local polynomial modeling, usually in the form of
a locally linear model \citep{Fan-Gijbels-1996}. Here, to prepare
for the estimation of locally linear coefficients, we augment the
design matrix with interactions between the covariates and location in two dimensions
\citep{Wang-2008b}. Let $\bm{X}=\left(\bm{X}_{1},\dots,\bm{X}_{n}\right)^{T}$
be the design matrix of observed covariate values. Then the augmented
local design matrix at location $\bm{s} = (u, v)^T$ is defined to be $\bm{Z}(\bm{s})=\left(\bm{X}\ \:\bm{L}(\bm{s}) \bm{X}\ \:\bm{M}(\bm{s}) \bm{X}\right),$
where $\bm{L}(\bm{s})=\text{diag}\{s_{i',1}-u\}_{i'=1}^{n}$ and
$\bm{M}(\bm{s})=\text{diag}\{s_{i',2}-v\}_{i'=1}^{n}$. 
The vector of augmented local coefficients at location $\bm{s}$ is defined to be $\bm{\zeta}(\bm{s})=\left(\bm{\beta}(\bm{s})^{T},\;\nabla_{u}\bm{\beta}(\bm{s})^{T},\;\nabla_{v}\bm{\beta}(\bm{s})^{T}\right)^{T}$,
where $\nabla_{u}\bm{\beta}(\bm{s})$ and $\nabla_{v}\bm{\beta}(\bm{s})$ denote the local gradients 
of the coefficient surfaces.

\subsection{Coefficient Estimation via Local Likelihood}


Let $\bm{\zeta}=\left(\bm{\zeta}(\bm{s}_{1}),\dots,\bm{\zeta}(\bm{s}_{n})\right)^{T}$
denote a matrix of the local coefficients at all observation locations
$\bm{s}_{1},\dots,\bm{s}_{n}$ and let $\left\{\bm{Z}(\bm{s})\right\}_{i}$ denote the
$i$th row of the matrix $\bm{Z}(\bm{s})$ as a column vector. The total log-likelihood of the observed
data is the sum of the log-likelihood of each individual observation:
\begin{align}
\ell\left(\bm{\zeta}\right)= & -(1/2)\sum_{i=1}^{n}\big( \log{\sigma^{2}}+\sigma^{-2} \left[ y_{i} - \{\bm{z}(\bm{s}_i)\}_i \bm{\zeta}(\bm{s}_{i}) \right] ^{2} \big).\label{eq-coefficients}
\end{align}

Since there are a total of $3(p+1)n+1$ parameters for $n$ observations,
the model is not identifiable and it is not possible to directly maximize
the total log-likelihood (\ref{eq-coefficients}). When the coefficient functions are smooth, though,
the coefficients $\bm{\zeta}(\bm{s})$ at location $\bm{s}$ can be
approximated by the coefficients $\bm{\zeta}(\bm{t})$ , where $\bm{t}$
is within some neighborhood of $\bm{s}$. This intuition is formalized
by the following local log-likelihood at location $\bm{s}\in\mathcal{D}$:
\begin{align}
\ell\left(\bm{\zeta}(\bm{s})\right)= & -(1/2)\sum_{i=1}^{n}K_{h}(\|\bm{s}-\bm{s}_{i}\|)\left[\log\sigma^{2}+\sigma^{-2}\left\{ y_{i}-\bm{z}_{i}^{T}\bm{\zeta}(\bm{s})\right\} ^{2}\right]\label{eq:local-log-likelihood}
\end{align}

where $\bm{Z}_{i}=\left\{\bm{Z}(\bm{s})\right\}_{i}$, $h$ is a bandwidth parameter, $\|\cdot\|$ is the $\ell_{2}$-norm,
and $K_{h}(\|\bm{s}-\bm{s}_{i}\|)$ for $i=1,\dots,n$ are local weights
from a kernel function. For instance, the Epanechnikov kernel is defined
as $K_{h}(\|\bm{s}_{i}-\bm{s}_{i'}\|)=h^{-2}K\left(h^{-1}\|\bm{s}_{i}-\bm{s}_{i'}\|\right)$
where $K(x)=(3/4)(1-x^{2})$ if $x<1$, and $0$ otherwise \citep{Samiuddin-el-Sayyad-1990}.

The local log-likelihood (\ref{eq:local-log-likelihood}) is maximized
to obtain an estimate $\tilde{\bm{\zeta}}(\bm{s})$ of the local coefficients
at $\bm{s}$. Let $\bm{W}\!(\bm{s})={\rm diag}\left\{ K_{h}(\|\bm{s}-\bm{s}_{i}\|)\right\} _{i'=1}^{n}$
denote a diagonal matrix of kernel weights. The local likelihood (\ref{eq:local-log-likelihood})
can be maximized by minimizing a locally weighted least squares: 
\begin{equation}
\mathcal{S}\left(\bm{\zeta}(\bm{s})\right)=(1/2)\left\{ \bm{Y}-\bm{Z}(\bm{s})\bm{\zeta}(\bm{s})\right\} ^{T}\bm{W}\!(\bm{s})\left\{ \bm{Y}-\bm{Z}(\bm{s})\bm{\zeta}(\bm{s})\right\} .\label{eq:local-sum-of-squares}
\end{equation}

The minimizer of (\ref{eq:local-sum-of-squares}) is the locally weighted least squares estimate
\begin{equation}
\tilde{\bm{\zeta}}(\bm{s})=\left\{ \bm{Z}(\bm{s})^{T}\bm{W}\!(\bm{s})\bm{Z}(\bm{s})\right\} ^{-1}\bm{Z}(\bm{s})^{T}\bm{W}\!(\bm{s})\bm{Y}.\label{eq:zeta-hat}
\end{equation}

By Theorem 3 of \citet{Sun-Yan-Zhang-Lu-2014}, for any given $\bm{s}$,
the estimated local coefficients $\tilde{\bm{\beta}}(\bm{s})=\left(\tilde{\zeta}_{1}(\bm{s}),\dots,\tilde{\zeta}_{p}(\bm{s})\right)^{T}$
converge in probability at the optimal rate of $O\left(n^{-1/3}\right)$
and are asymptotically normally distributed. The bias of the local
coefficient estimates is proportional to the second derivatives of
the true coefficient functions.

\section{Local Variable Selection with LAGR\label{sec:lagr-gaussian}}

\subsection{LAGR Penalized Local Likelihood}

Estimating the local coefficients by (\ref{eq:zeta-hat}) has traditionally
relied on variable selection \emph{a priori}; that is, a set of covariates
is pre-determined. Here we develop a new method of penalized regression
to simultaneously select covariates locally and estimate the corresponding
local coefficients. For this purpose, each raw covariate is grouped
with its covariate-by-location interactions. That is, $\bm{\zeta}_{(j)}(\bm{s})=\left(\beta_{j}(\bm{s}),\;\;\nabla_{u}\beta_{j}(\bm{s}),\;\;\nabla_{v}\beta_{j}(\bm{s})\right)^{T}$
for $j=1,\dots,p$. The proposed penalty is akin to the adaptive
group Lasso \citep{Yuan-Lin-2006,Wang-Leng-2008}. By the mechanism
of the adaptive group Lasso, covariates within the same group are
included in or dropped from the model together. The intercept group
is left unpenalized.

To select and estimate the local coefficients at location $\bm{s}$,
we minimize a penalized local sum of squares at location $\bm{s}$:
\begin{align}
\mathcal{J}\left(\bm{\zeta}(\bm{s})\right) & =\mathcal{S}\left(\bm{\zeta}(\bm{s})\right)+\mathcal{P}\left(\bm{\zeta}(\bm{s})\right),\label{eq:penalized-least-squares}
\end{align}

where $\mathcal{S}\left(\bm{\zeta}\left(\bm{s}\right)\right)$ is
the locally weighted least squares defined in (\ref{eq:local-sum-of-squares}),
$\mathcal{P}\left(\bm{\zeta}(\bm{s})\right)=\sum_{j=1}^{p}\phi_{j}(\bm{s})\|\bm{\zeta}_{(j)}(\bm{s})\|$
is a local adaptive grouped regularization (LAGR) penalty. The LAGR
penalty for the $j$th group of coefficients at location $\bm{s}$
is $\phi_{j}(\bm{s})=\lambda_{n}\|\tilde{\bm{\zeta}}_{(j)}(\bm{s})\|^{-\gamma}$,
where $\lambda_{n}>0$ is a local tuning parameter applied to all
coefficients at location $\bm{s}$, $\tilde{\bm{\zeta}}_{(j)}(\bm{s})$
is a vector of unpenalized local coefficients for the $j$th covariate
and its interactions with location from (\ref{eq:zeta-hat}), and $\gamma>1$.

Minimization of (\ref{eq:penalized-least-squares}) is by blockwise
coordinate descent, where each block is a covariate group (one raw
covariate and its interactions with location). A companion software package for estimating
$\bm{\zeta}(\bm{s})$ will be made available in \texttt{R} \citep{R-Core-2014}.

\subsection{Oracle Property\label{sub:oracle-properties}}

At location $\bm{s}$, let there be $p_{0}(\bm{s})<p$ covariates $\bm{X}_{(a)}(\bm{s})$
with nonzero local regression coefficients, denoted $\bm{\beta}_{(a)}(\bm{s})\ne\bm{0}$.
Without loss of generality, assume the indices of these covariates
are $1,\dots,p_{0}(\bm{s})$. The remaining $p-p_{0}(\bm{s})$ covariates $\bm{X}_{(b)}(\bm{s})$
have coefficients equal to zero, denoted $\bm{\beta}_{(b)}(\bm{s})=\bm{0}$.
Denote by $a_{n}=\max\left\{ \phi_{j}(\bm{s}),j\le p_{0}(\bm{s})\right\} $
the largest penalty applied to a covariate group whose true coefficient
norm is nonzero, and by $b_{n}=\min\left\{ \phi_{j}(\bm{s}),j>p_{0}(\bm{s})\right\} $
the smallest penalty applied to a covariate group whose true coefficient
norm is zero. Let $\bm{Z}_{(k)}(\bm{s})$ be the augmented design
matrix for covariate group $k$, and let $\bm{Z}_{(\mhyphen k)}(\bm{s})$
be the augmented design matrix for all the data except covariate group
$k$. Similarly, let $\bm{\zeta}_{(k)}(\bm{s})$ be the augmented
coefficients for covariate group $k$ and $\bm{\zeta}_{(\mhyphen k)}(\bm{s})$
be the augmented coefficients for all covariate groups except $k$.
Let $\nabla\zeta_{j}(\bm{s})=\left(\nabla_{u}\zeta_{j}(\bm{s}),\nabla_{v}\zeta_{j}(\bm{s})\right)^{T}$
and $\nabla^{2}\zeta_{j}(\bm{s})=\left(\begin{array}{cc}
\nabla_{uu}^{2}\zeta_{j}(\bm{s}) & \nabla_{uv}^{2}\zeta_{j}(\bm{s})\\
\nabla_{vu}^{2}\zeta_{j}(\bm{s}) & \nabla_{vv}^{2}\zeta_{j}(\bm{s})
\end{array}\right)$. Let $\kappa_{0}=\int_{\mathbb{R}^{2}}K(\|\bm{s}\|)d\bm{s}$, $\kappa_{2}=\int_{\mathbb{R}^{2}}[(1,0)\bm{s}]^{2}K(\|\bm{s}\|)d\bm{s}=\int_{\mathbb{R}^{2}}[(0,1)\bm{s}]^{2}K(\|\bm{s}\|)d\bm{s}$,
and $\nu_{0}=\int_{\mathbb{R}^{2}}K^{2}(\|\bm{s}\|)d\bm{s}$. Finally,
let $\xrightarrow{p}$ and $\xrightarrow{d}$ denote convergence in
probability and distribution, respectively, as $n\to\infty$.

Assume the following regularity conditions.
\begin{itemize}
\item[(C.1)] The kernel function $K(\cdot)$ is bounded, positive, symmetric,
and Lipschitz continuous on $\mathbb{R}$, and has a bounded support.
\item[(C.2)] The coefficient functions $ $$\beta_{j}(\cdot)$ for $j=1,\dots,p$
have continuous second-order partial derivatives at $\bm{s}$.
\item[(C.3)] The covariates $\bm{X}(\bm{s}_{1}),\dots,\bm{X}(\bm{s}_{n})$ are
random vectors that are independent of $\varepsilon_{1},\dots,\varepsilon_{n}$.
Also $\bm{\Psi}(\bm{s})=E\left\{ \bm{X}(\bm{s})\bm{X}(\bm{s})^{T}|\bm{s}\right\} $
and $\bm{\Psi}_{(a)}(\bm{s})=E\left\{ \bm{X}_{(a)}(\bm{s})\bm{X}_{(a)}(\bm{s})^{T}|\bm{s}\right\} $
are positive-definite and differentiable at location $\bm{s}$.
\item[(C.4)] $E\left\{ \left|\bm{X}(\bm{s})\right|^{3}|\bm{s}\right\} $ and $E\left\{ Y(\bm{s})^{4}|\bm{X}(\bm{s}),\bm{s}\right\} $
are continuous at a given location $\bm{s}$.
\item[(C.5)] The observation locations $\{\bm{s}_{i}\}$ are a sequence of
design points on a bounded compact support $\mathcal{S}$. Further,
there exists a positive joint density function $f(\cdot)$ satisfying
a Lipschitz condition such that 
\[
\sup_{\bm{s}\text{\ensuremath{\in}}\mathcal{S}}\left|n^{-1}\sum_{i=1}^{n}\left[r(\bm{s}_{i})K_{h}(\|\bm{s}_{i}-\bm{s}\|)\right]-\int r(\bm{t})K_{h}(\|\bm{t}-\bm{s}\|)f(\bm{t})d\bm{t}\right|=O(h)
\]
where $f(\cdot)$ is bounded away from zero on $\mathcal{S}$, $r(\cdot)$
is any bounded continuous function, and $K_{h}(\cdot)=K(\cdot/h)/h^{2}$.
\item[(C.6)] $h=O\left(n^{-1/6}\right)$.
\item[(C.7)] $h^{-1}n^{-1/2}a_{n}\xrightarrow{p}0$ and $hn^{-1/2}b_{n}\xrightarrow{p}\infty$.
\end{itemize}
Conditions (C.1)--(C.4) are common in the literature on nonparametric
estimation, for instance see conditions (1)--(3) of \citet{Sun-Yan-Zhang-Lu-2014}
and conditions (5) and (6) of \citet{Cai-Fan-Li-2000}. However, the
covariates $\bm{X}(\bm{s}_{1}),\dots,\bm{X}(\bm{s}_{n})$ were assumed
to be $iid$ in \citet{Sun-Yan-Zhang-Lu-2014}, which is not required
here. The existence of of $\bm{\Psi}(\cdot)$ is needed for the existence
of the limiting distribution of $\hat{\bm{\beta}}(\bm{s})$; its differentiability
is used in the Taylor's expansions. Condition (C.4) is used when bounding
the remainder term in the Taylor's expansions. Condition (C.5) is
the same as condition (4) of \citet{Sun-Yan-Zhang-Lu-2014}.
Under condition (C.6), the coefficient estimates attain the optimal
rate of convergence for bivariate nonparametric regression. Condition
(C.7) is needed for establishing the oracle properties, and is a refinement
of the condition for the adaptive group Lasso \citep{Wang-Leng-2008}.

In particular, satisfying (C.7) implies an additional restriction
on $\gamma$, the unpenalized group norm exponent in the LAGR penalty.
Under condition (C.7), the local penalty tends to zero on covariates
with true nonzero coefficients and to infinity on covariates with
true zero coefficients. By (C.7), $h^{-1}n^{-1/2}\lambda_{n}\to0$
for all $j\le p_{0}(\bm{s})$ and $hn^{-1/2}\lambda_{n}\|\bm{\zeta}_{(j)}(\bm{s})\|^{-\gamma}\to\infty$
for all $j>p_{0}(\bm{s})$. We require that $\lambda_{n}$ satisfy both assumptions.
Suppose $\lambda_{n}=n^{\alpha}$. Since $h=O\left(n^{-1/6}\right)$
and $\|\tilde{\bm{\zeta}}_{(p)}(\bm{s})\|=O\left(h^{-1}n^{-1/2}\right)$,
it follows that $h^{-1}n^{-1/2}\lambda_{n}=O\left(n^{-1/3+\alpha}\right)$
and $hn^{-1/2}\lambda_{n}\|\tilde{\bm{\zeta}}_{(p)}(\bm{s})\|^{-\gamma}=O\left(n^{-2/3+\alpha+\gamma/3}\right)$.
Thus, $\left(2-\gamma\right)/3<\alpha<1/3$, which can only be satisfied
for $\gamma>1$.
\begin{thm}[Asymptotic normality]
\label{theorem:normality}  Under (C.1)--(C.8),
\begin{gather*}
\left\{ f(\bm{s})h^{2}n\right\} ^{1/2}\left[\hat{\bm{\beta}}_{(a)}(\bm{s})-\bm{\beta}_{(a)}(\bm{s})-(2\kappa_{0})^{-1}\kappa_{2}h^{2}\left\{ \nabla_{uu}^{2}\bm{\beta}_{(a)}(\bm{s})+\nabla_{vv}^{2}\bm{\beta}_{(a)}(\bm{s})\right\} \right]\\
\xrightarrow{d}N\left(0,\kappa_{0}^{-2}\nu_{0}\sigma^{2}\Psi_{(a)}(\bm{s})^{-1}\right),
\end{gather*}
where $\left\{ \nabla_{uu}^{2}\bm{\beta}_{(a)}(\bm{s})+\nabla_{vv}^{2}\bm{\beta}_{(a)}(\bm{s})\right\} =\left(\nabla_{uu}^{2}\bm{\beta}_{1}(\bm{s})+\nabla_{vv}^{2}\bm{\beta}_{1}(\bm{s}),\dots,\nabla_{uu}^{2}\bm{\beta}_{p_{0}}(\bm{s})+\nabla_{vv}^{2}\bm{\beta}_{p_{0}}(\bm{s})\right)^{T}$.
\end{thm}

\begin{thm}[Selection consistency]
\label{theorem:selection} Under (C.1)--(C.8), for $j>p_{0}(\bm{s})$,
\[
P\left\{ \|\hat{\bm{\zeta}}_{(j)}(\bm{s})\|=\bm{0}\right\} \to1.
\]
\end{thm}
Theorem \ref{theorem:normality} indicates that the LAGR estimates
for true nonzero coefficients have the same asymptotic distribution
as a local regression model where the true nonzero coefficients are
known in advance. Further, by Theorem \ref{theorem:selection}, the
LAGR estimates of true zero coefficients tend to zero with probability
one. Together, local variable selection and local coefficient estimation by LAGR
have the oracle property. The technical proofs of Theorems \ref{theorem:normality}
and \ref{theorem:selection} are given in Appendix A.

\subsection{Tuning Parameter Selection}

In practical application, it is necessary to select the LAGR tuning
parameter $\lambda_{n}$ for each local model. A popular approach
in other Lasso-type problems is to select the tuning parameter that
maximizes a criterion that approximates the expected log-likelihood
of a new, independent data set drawn from the same distribution. This
is the framework of Mallows' ${\rm C_p}$, Stein's unbiased risk estimate (SURE)
and Akaike's information criterion (AIC) \citep{Mallows-1973,Stein-1981,Akaike-1973}.

These criteria use a so-called covariance penalty to estimate the
bias due to using the same data set to select a model and to estimate
its parameters \citep{Efron:2004a}. We adopt the approximate degrees
of freedom for the adaptive group Lasso from \citet{Yuan-Lin-2006}
and minimize the AIC to select the tuning parameter $\lambda_{n}$. That is, let

\begin{align*}
\hat{{\rm df}}(\lambda_{n};\bm{s})= & \sum_{j=1}^{p}I\left(\|\hat{\bm{\zeta}}(\lambda_{n};\bm{s})\|>0\right) + d \sum_{j=1}^{p}\|\hat{\bm{\zeta}}(\lambda_{n};\bm{s})\|\|\tilde{\bm{\zeta}}(\bm{s})\|^{-1},\\
\text{AIC}(\lambda_{n};\bm{s})= & \sum_{i=1}^{n}K_{h}(\|\bm{s}-\bm{s}_{i}\|)\sigma^{-2}\left\{ y_{i}-\bm{z}_{i}^{T}\hat{\bm{\zeta}}(\lambda_{n};\bm{s})\right\} ^{2}+2\hat{{\rm df}}(\lambda_{n};\bm{s}),\\
\end{align*}
where $I\left(\cdot\right)$ is the indicator function, $d$ is the dimension of the effect-modifying parameter, and the local
coefficient estimate is written $\hat{\bm{\zeta}}(\lambda_{n};\bm{s})$
to emphasize that it depends on the tuning parameter. Here, $d=2$. More general dimensionalilty is discussed in Section \ref{section-conclusion}.

\section{Simulation Study\label{sec:simulations}}

\subsection{Simulation Setup}

A simulation study was conducted to assess the performance of the
method described in Sections \ref{sec:vcr}--\ref{sec:lagr-gaussian}.
Data were simulated on the domain $[0,1]^{2}$, which was divided
into a $20\times20$ grid. Each of $p=5$ covariates $X_{1},\dots,X_{5}$
was simulated by a Gaussian random field (GRF) with mean zero, nugget variance $0.2$, and exponential
covariance $\text{Cov}\left(X_{ij},X_{i'j}\right)=\sigma_{x}^{2}\exp\left(-\tau_{x}^{-1}\delta_{ii'}\right)$
where $\sigma_{x}^{2}=1$ is the variance, $\tau_{x}=0.1$ is the
range parameter, and $\delta_{ii'}=\|\bm{s}_{i}-\bm{s}_{i'}\|$. Correlation
was induced between the covariates by multiplying the design matrix
$\bm{X}$ by $\bm{R}$, where $\bm{R}$ is the Cholesky decomposition
of the covariance matrix $\bm{\Sigma}=\bm{R}^{T}\bm{R}$. The covariance
matrix $\bm{\Sigma}$ is a $5\times5$ matrix that has ones on the
diagonal and $\rho$ for all off-diagonal entries, where $\rho$ is
the between-covariate correlation. 

The simulated response was $y_{i}=\bm{x}_{i}^{T}\bm{\beta}(\bm{s}_{i})+\varepsilon_{i}$
for $i=1,\dots,n$ where $n=400$ and the $\varepsilon_{i}$'s were
iid Gaussian with mean zero and variance $\sigma_{\varepsilon}^{2}$.
The coefficients $\beta_1(\bm{s}), \dots, \beta_3(\bm{s})$ were generated by GRFs, and the fourth coefficient was $\beta_4(\bm{s}) \equiv 0$.
The GRFs for generating $\beta_1(\bm{s}), \dots, \beta_3(\bm{s})$ had mean zero, no nugget variance, and exponential
covariance $\text{Cov}\left(\beta_j(\bm{s}_i), \beta_j(\bm{s}_{i'})\right)=\sigma_j^2 \exp\left(-\tau_{\beta}^{-1}\delta_{ii'}\right)$
where $\tau_{\beta}=1$ is the range parameter. The scale of the coefficient surface $\beta_j(\bm{s})$ was set via the variance $\sigma_j^2$, and the values used in the simulations were $\sigma_1^2 = 10, \sigma_2^2=1, \sigma_3^2=0.1$.
These values were chosen so that the the covariates $X_1, X_2, X_3$ would have progressively less influence on the response.
The coefficient values $\beta_1(\bm{s}), \dots, \beta_3(\bm{s})$ generated in this way are plotted in Figure \ref{fig:simulation-coefficients}.

\begin{figure}
    \includegraphics[width=\textwidth]{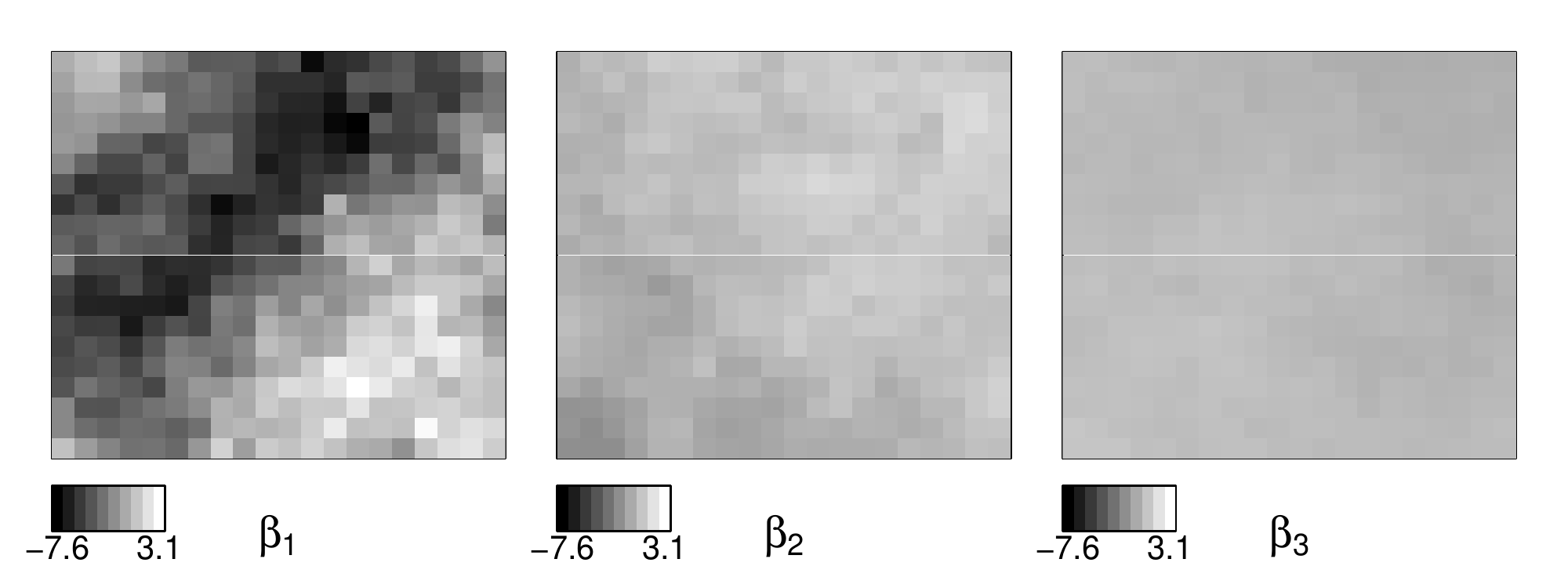}
    \caption{Left to right, the values used for coefficients $\beta_1(\bm{s}), \dots, \beta_3(\bm{s})$ in the simulation study. \label{fig:simulation-coefficients}}
\end{figure}

Two parameters were varied to produce six simulation settings.
Data were simulated with low ($\rho=0$), medium ($\rho=0.5$), or high ($\rho=0.9$)
correlation between the covariates, and with
low ($\sigma_{\varepsilon}=0.5$) or high ($\sigma_{\varepsilon}=1$)
variance for the random error term. Each setting was used to generate five data sets consisting of 400 observations each.
For each data set, estimates were made of the coefficients for three different sample sizes $N$: the full 400 observations, and subsets generated by sampling 100 or 200 unique observations uniformly from the data set.
The coefficients were estimated via LAGR and via a VCR model without variable selection as in Section \ref{sec:lagr-gaussian}.
For both estimation methods, the bandwidth parameter was $h=(3/2)N^{-1/6} - 0.36$ with a nearest neighbors type bandwidth, meaning the kernel bandwidth was adjusted at each location $\bm{s}_i$ to achieve the ratio $\sum_{i'=1}^n w_{ii'}/n = h$.

\begin{table}
	\centering
	\begin{tabular}{|ccc|cc|cc|cc|cc|}
		\hline
		\multicolumn{3}{|c|}{\begin{tabular}[c]{@{}c@{}}Simulation\\settings\end{tabular}} & \multicolumn{2}{c|}{\begin{tabular}[c]{@{}c@{}}MISE\\$\hat{\beta}_1$\end{tabular}} & \multicolumn{2}{c|}{\begin{tabular}[c]{@{}c@{}}MISE\\$\hat{\beta}_2$\end{tabular}} & \multicolumn{2}{c|}{\begin{tabular}[c]{@{}c@{}}MISE\\$\hat{\beta}_3$\end{tabular}} & \multicolumn{2}{c|}{\begin{tabular}[c]{@{}c@{}}MISE\\$\hat{\beta}_4$\end{tabular}} \\
		$n$ & $\rho$ & $\sigma_{\varepsilon}$ & LAGR & VCR & LAGR & VCR & LAGR & VCR & LAGR & VCR \\
		\hline 

  \multirow{6}{*}{100} & \multirow{2}{*}{0} & 0.5 & 2.16 & \textbf{2.15} & \textbf{0.35} & 0.36 & \textbf{0.18} & 0.24 & \textbf{0.15} & 0.29 \\ 
    &  & 1.0 & 2.19 & \textbf{2.14} & \textbf{0.38} & 0.38 & \textbf{0.17} & 0.28 & \textbf{0.16} & 0.35 \\ 
    & \multirow{2}{*}{0.5} & 0.5 & \textbf{2.36} & 2.47 & 0.40 & \textbf{0.35} & \textbf{0.19} & 0.27 & \textbf{0.26} & 0.48 \\ 
    &  & 1.0 & \textbf{2.25} & 2.48 & 0.44 & \textbf{0.39} & \textbf{0.18} & 0.34 & \textbf{0.24} & 0.58 \\ 
    & \multirow{2}{*}{0.9} & 0.5 & \textbf{3.00} & 4.90 & \textbf{0.68} & 1.16 & \textbf{0.35} & 1.07 & \textbf{0.70} & 2.22 \\ 
    &  & 1.0 & \textbf{2.77} & 5.18 & \textbf{0.61} & 1.35 & \textbf{0.38} & 1.37 & \textbf{0.53} & 2.71 \\ 
   \hline \multirow{6}{*}{200} & \multirow{2}{*}{0} & 0.5 & 1.75 & \textbf{1.72} & 0.20 & \textbf{0.18} & \textbf{0.09} & 0.15 & \textbf{0.03} & 0.10 \\ 
    &  & 1.0 & 1.79 & \textbf{1.78} & 0.27 & \textbf{0.21} & \textbf{0.11} & 0.22 & \textbf{0.05} & 0.13 \\ 
    & \multirow{2}{*}{0.5} & 0.5 & 1.80 & \textbf{1.75} & 0.25 & \textbf{0.22} & \textbf{0.12} & 0.23 & \textbf{0.05} & 0.15 \\ 
    &  & 1.0 & 1.84 & \textbf{1.83} & 0.32 & \textbf{0.28} & \textbf{0.18} & 0.34 & \textbf{0.07} & 0.21 \\ 
    & \multirow{2}{*}{0.9} & 0.5 & \textbf{2.19} & 2.37 & \textbf{0.43} & 0.76 & \textbf{0.36} & 0.98 & \textbf{0.24} & 0.75 \\ 
    &  & 1.0 & \textbf{2.25} & 2.66 & \textbf{0.52} & 1.10 & \textbf{0.57} & 1.48 & \textbf{0.32} & 1.01 \\ 
   \hline \multirow{6}{*}{400} & \multirow{2}{*}{0} & 0.5 & 1.34 & \textbf{1.33} & 0.18 & \textbf{0.15} & \textbf{0.06} & 0.06 & \textbf{0.02} & 0.05 \\ 
    &  & 1.0 & 1.37 & \textbf{1.35} & 0.22 & \textbf{0.17} & \textbf{0.08} & 0.08 & \textbf{0.02} & 0.05 \\ 
    & \multirow{2}{*}{0.5} & 0.5 & 1.37 & \textbf{1.35} & 0.20 & \textbf{0.18} & \textbf{0.07} & 0.09 & \textbf{0.03} & 0.08 \\ 
    &  & 1.0 & 1.40 & \textbf{1.39} & 0.25 & \textbf{0.21} & \textbf{0.09} & 0.13 & \textbf{0.03} & 0.09 \\ 
    & \multirow{2}{*}{0.9} & 0.5 & \textbf{1.55} & 1.66 & \textbf{0.41} & 0.47 & \textbf{0.16} & 0.36 & \textbf{0.15} & 0.40 \\ 
    &  & 1.0 & \textbf{1.57} & 1.84 & \textbf{0.44} & 0.64 & \textbf{0.17} & 0.54 & \textbf{0.14} & 0.46 \\

	\hline
	\end{tabular}
	\caption{For each setting as a combination of sample size $n$, cross-covariate correlation $\rho$, and error standard deviation $\sigma_{\varepsilon}$, the mean integrated squared error (MISE) of the coefficient estimates, averaged over five independent data sets for each simulation setting.
	The MISE of $\hat{\beta}_1, \dots, \hat{\beta}_4$ from estimation by local adaptive grouped regularization (LAGR) is compared to that from estimation by locally linear regression without selection (VCR).
	\textbf{Highlighting} indicates whether LAGR or VCR produced the smaller MISE for each coefficient surface under each simulation setting.}
	\label{tab:mise}
\end{table}

\subsection{Simulation Results}

The mean integrated squared error (MISE) of the coefficient surface estimates are in Table \ref{tab:mise}, where the MISE is calculated as
$\text{MISE}(\beta_j) = N^{-1} \sum_{i=1}^N \{ \hat{\beta}_j(\bm{s}_i) - \beta_j(\bm{s}_i) \}^2$.
The results in Table \ref{tab:mise} are averaged over the five independent data sets for each simulation setting. In general, the coefficients estimated by LAGR were more accurate in terms of MISE than those estimated by VCR. 
Although the frequency with which MISE was smaller under LAGR than under VCR for estimating $\beta_1$ and $\beta_2$ was $8$ of $18$ cases each, 
the improvement by MISE for LAGR over VCR was greater for covariates with smaller influence, with LAGR producing the smaller MISE for $\beta_3$ and $\beta_4$ in every case.
In no case was the MISE for LAGR more than $27\%$ greater than for VCR. 
The MISE for estimating $\beta_4$ with $\rho=0.9$, $\sigma_{\varepsilon}=1.0$, and $n=100$ setting was $5$ times greater for VCR than for LAGR, and under the other simulation settings the greatest improvement for LAGR over VCR tended to be a $2-3$ times reduction in MISE.

With other factors held constant, the MISE for estimating the coefficients tended to be smaller for less influential covariates and under larger sample sizes.
On the other hand, the MISE tended to increase with high error variance or increasing correlation between covariates.
In terms of MISE, the improvement from estimation by LAGR over VCR was greater for settings with smaller sample sizes, higher correlation between covariates, and greater error variance.
In fact, estimation by LAGR was always more accurate than estimation by VCR under high cross-covariate correlation $(\rho=0.9)$.

The frequencies of exact zeros among the estimates of each coefficient for each simulation setting are in Table \ref{tab:zerofreq}.
The frequency of exact zeros in the coefficient estimates generally increased as covariates grew less influential.
In particular, the estimates $\hat{\beta}_1$ were almost never exactly zero, while the estimates $\hat{\beta}_3$ and $\hat{\beta}_4$ were exactly zero more often than not.
Exact zero coefficient estimates were generally more frequent under smaller sample sizes, greater error variance, and greater cross-covariate correlation. 
Under high cross-covariate correlation, the frequency of exact zero estimates was roughly equal (and in the neighborhood of $75\%$) for $\beta_2$, $\beta_3$, and $\beta_4$, which indicates that under high cross-covariate correlation, LAGR tended to include only the most influential covariate.

\begin{table}
	\centering
	\begin{tabular}{|ccc|cccc|}
		\hline
		\multicolumn{3}{|c|}{\begin{tabular}[c]{@{}c@{}}Simulation\\settings\end{tabular}} &  \multicolumn{4}{c|}{\begin{tabular}[c]{@{}c@{}}Zero\\frequency\end{tabular}} \\
		$n$ & $\rho$ & $\sigma_{\varepsilon}$ & $\hat{\beta}_1$ & $\hat{\beta}_2$ & $\hat{\beta}_3$ & $\hat{\beta}_4$ \\
		\hline

  \multirow{6}{*}{100} & \multirow{2}{*}{0} & 0.5 & 0.00 & 0.40 & 0.67 & 0.76 \\ 
    &  & 1.0 & 0.00 & 0.57 & 0.72 & 0.80 \\ 
    & \multirow{2}{*}{0.5} & 0.5 & 0.00 & 0.47 & 0.68 & 0.75 \\ 
    &  & 1.0 & 0.01 & 0.65 & 0.77 & 0.79 \\ 
    & \multirow{2}{*}{0.9} & 0.5 & 0.00 & 0.76 & 0.79 & 0.78 \\ 
    &  & 1.0 & 0.02 & 0.84 & 0.79 & 0.76 \\ 
   \hline \multirow{6}{*}{200} & \multirow{2}{*}{0} & 0.5 & 0.00 & 0.28 & 0.68 & 0.83 \\ 
    &  & 1.0 & 0.00 & 0.39 & 0.66 & 0.84 \\ 
    & \multirow{2}{*}{0.5} & 0.5 & 0.00 & 0.36 & 0.66 & 0.81 \\ 
    &  & 1.0 & 0.00 & 0.48 & 0.69 & 0.84 \\ 
    & \multirow{2}{*}{0.9} & 0.5 & 0.03 & 0.67 & 0.74 & 0.83 \\ 
    &  & 1.0 & 0.04 & 0.71 & 0.69 & 0.84 \\ 
   \hline \multirow{6}{*}{400} & \multirow{2}{*}{0} & 0.5 & 0.00 & 0.18 & 0.56 & 0.74 \\ 
    &  & 1.0 & 0.00 & 0.31 & 0.64 & 0.82 \\ 
    & \multirow{2}{*}{0.5} & 0.5 & 0.00 & 0.24 & 0.62 & 0.73 \\ 
    &  & 1.0 & 0.00 & 0.36 & 0.69 & 0.80 \\ 
    & \multirow{2}{*}{0.9} & 0.5 & 0.02 & 0.61 & 0.77 & 0.73 \\ 
    &  & 1.0 & 0.02 & 0.68 & 0.75 & 0.80 \\

	\hline
	\end{tabular}
	\caption{For each setting as a combination of sample size $n$, cross-covariate correlation $\rho$, and error standard deviation $\sigma_{\varepsilon}$, the frequency of exact zeroes in the estimates of $\hat{\beta}_1,\dots,\hat{\beta}_4$ as estimated by local adaptive grouped regularizaton.}
	\label{tab:zerofreq}
\end{table}

\section{Data Example\label{sec:example}}

The proposed LAGR estimation method was applied to estimate the coefficients
in a VCR model for the price of homes
in Boston based on data from the 1970 U.S. census \citep{Harrison-Rubinfeld-1978,Gilley-Pace-1996,Pace-Gilley-1997}.
The data are the median price of homes sold in 506 census tracts (MEDV),
along with the potential covariates CRIM (the per-capita crime rate
in the tract), RM (the mean number of rooms for houses sold in the
tract), RAD (an index of how accessible the tract is from Boston's
radial roads), TAX (the property tax per \$10,000 of property value),
and LSTAT (the percentage of the tract's residents who are considered
``lower status''). With the Epanechnikov kernel, the nearest neighbors type bandwidth
was set to $h=0.26$.

\begin{figure}

\includegraphics[width=\maxwidth]{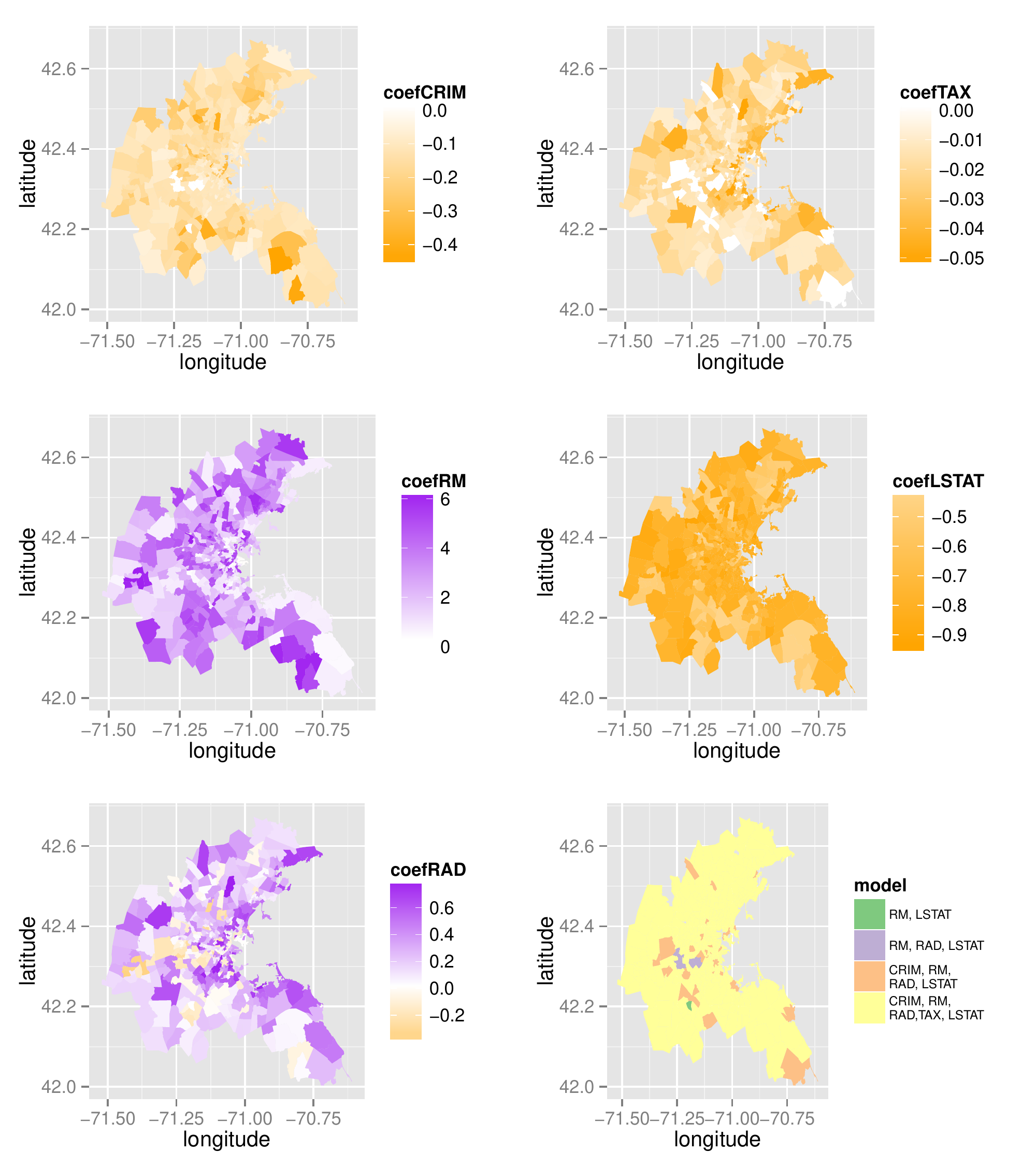} 
\caption{A varying coefficient regression model for the median house price in each census tract in Boston in 1970, estimated by local adaptive grouped regularization.
In the left column are the estimated coefficients for covariates CRIM (per-capita crime rate), RM (mean number of rooms per house), and RAD (an index of access to radial roads.
In the right column are the estimated coefficients for covariates TAX (property tax per $\$10,000$) and LSTAT (proportion of residents who are ``lower status"), and a map indicating which covariates were estimated to have nonzero coefficients in each census tract.
\label{fig:boston-lagr-coefs}}
\end{figure}

The estimates of the local coefficients are plotted in the first five panels of Figure \ref{fig:boston-lagr-coefs} and are summarized in Table \ref{tab:boston-coefs-lagr}.
The estimated coefficients of CRIM and LSTAT were everywhere negative or exactly zero, suggesting that the crime rate and proportion of ``lower-status" individuals were associated with a lower median house price. 
Meanwhile, the coefficient of RM was everywhere estimated to be positive, so the more rooms in the average house was everywhere associated with a higher median house price.
The coefficient of TAX was negative in most census tracts, but was estimated to be exactly zero in 50 tracts, indicating no discernable effect of the property tax rate on house prices in those tracts.
The coefficient of RAD is positive in some areas and negative
in others. This indicates that there are parts of Boston where access
to radial roads is associated with a greater median house price and
parts where it is associated with a lesser median house price.
The bottom right panel of Figure \ref{fig:boston-lagr-coefs} shows which covariates were estimated to have a nonzero coefficient in each tract.
There were 471 tracts where all LAGR estimated that all the covariates had a nonzero coefficient, 43 tracts where all covariates except for TAX were estimated to have nonzero coefficients, six tracts where the coefficients of CRIM and TAX were estimated to be zero, and one tract where the coefficients of CRIM, RAD, and LSTAT were estimated to be zero.

\ifthenelse{\boolean{thesis}}{
An interesting result from the data example is the apparent relationship between the property tax rate and the coefficient of the TAX variable.
Of the $506$ Boston census tracts, there were $50$ where the estimated coefficient of the TAX variable was zero.
Four of those $50$ tracts ($8\%$) were locations where the property tax rate was at least $\$666$ per $\$10000$ (which encompasses the two highest property tax rates).
In total there were $137$ tracts where the property tax rate was at least $\$666$ per $\$10000$, which was $27\%$ of all tracts.
This suggests that the coefficient of the TAX variable was less likely to be zero in tracts with the highest property tax rates than in other tracts.
A rank sum test was used to test the null hypothesis that tracts where the coefficient of the TAX variable was estimated to be zero occurred independently of the property tax rate.
The alternative considered was that the tracts with an estimated zero TAX coefficient were clustered among the tracts with a lower property tax rate.
Of $1000$ uniform samples of size $50$ from the ranks of the property tax rates, only four had a smaller sum than that observed for the tracts where the TAX coefficient was estimated to be zero.
This provides evidence that the property tax rate was more likely to have no apparent effect on the median house price in tracts where the property tax was lower.}{}

\begin{table}
\centering
\begin{tabular}{|c|rcc|}
  \hline
 Covariate & Mean & \begin{tabular}{c}Standard \\ dev. \end{tabular} & \begin{tabular}{c}Zero coef. \\ count \end{tabular}  \\ 
  \hline
CRIM & -0.15 & 0.07 & 7 \\ 
  RM & 2.56 & 1.68 & 0 \\ 
  RAD & 0.21 & 0.25 & 1 \\ 
  TAX & -0.02 & 0.01 & 50 \\ 
  LSTAT & -0.73 & 0.13 & 0 \\ 
   \hline
\end{tabular}
\caption{The mean, standard deviation, and count of zeros among the estimates of the local coefficients in a model for the median house price in 506 census tracts in Boston, with coefficients selected and fitted by local adaptive grouped regularization. The covariates are CRIM (per capita crime rate in the census tract), RM (average number of rooms per home sold in the census tract), RAD (an index of the tract's access to radial roads), TAX (property tax per USD10,000 of property value), and LSTAT (percentage of the tract's residents who are considered ``lower status").}
\label{tab:boston-coefs-lagr} 
\end{table}

\section{Extension to Generalized Linear Regression\label{sec:lagr-gllm}}

\subsection{Local GLM and Local Quasi-likelihood Estimation}

Generalized linear models (GLMs) extend the linear regression model
to a response variable following any distribution in the exponential
family \citep{McCullagh-Nelder-1989}. As is the case for the local
linear regression model, we now consider local GLM coefficients as
smooth functions of location \citep{Cai-Fan-Li-2000}. Suppose the
response variable $Y$ is from an exponential family distribution
with $E\left\{ y(\bm{s})|\bm{x}(\bm{s})\right\} =\mu(\bm{s})=b'\left(\theta(\bm{s})\right)$,
$\theta(\bm{s})=(g\circ b')^{-1}\left(\eta(\bm{s})\right)$, $\eta(\bm{s})=\bm{x}(\bm{s})^{T}\bm{\beta}(\bm{s})=g\left(\mu(\bm{s})\right)$,
$\text{\text{Var}}\left\{ y(\bm{s})|\bm{x}(\bm{s})\right\} =b''\left(\theta(\bm{s})\right)$,
and link function $g(\cdot)$. Then the probability density is 

\[
f\left(y(\bm{s})|\bm{x}(\bm{s}),\theta(\bm{s})\right)=c\left(y(\bm{s})\right)\times\exp\left\{ \theta(\bm{s})y(\bm{s})-b\left(\theta(\bm{s})\right)\right\} .
\]

If $g^{-1}(\cdot)=b'(\cdot)$, then the composition $(g\circ b')(\cdot)$
is the identity function. This particular $g$ is called the canonical
link. Assuming the canonical link, all that is required is to specify
the mean-variance relationship via the variance function, $V\left(\mu(\bm{s})\right)$.
Then the local coefficients can be estimated by maximizing the local
quasi-likelihood 

\begin{align}
\mathcal{\ell}^{*}\left(\bm{\zeta}(\bm{s})\right) & = \sum_{i=1}^{n}K_{h}\left(\|\bm{s}-\bm{s}_{i}\|\right)Q\left(g^{-1}\left(\bm{z}_{i}^{T}\bm{\zeta}(\bm{s})\right),Y_{i}\right).\label{eq:local-quasi-likelihood}
\end{align}

The local quasi-likelihood (\ref{eq:local-quasi-likelihood}) generalizes
the local log-likelihood (\ref{eq:local-log-likelihood}) that was
used to estimate coefficients in the local linear regression. The
local quasi-likelihood (\ref{eq:local-quasi-likelihood}) is concave,
and is defined in terms of its derivative, the local quasi-score function
$\left(\partial/\partial\mu\right)Q(\mu,y)=(y-\mu)\{V(\mu)\}^{-1}$.
The local quasi-likelihood is maximized by setting the local quasi-score
function to zero:

\begin{align}
(\partial/\partial\bm{\zeta})\mathcal{\ell}^{*}\left(\hat{\bm{\zeta}}(\bm{s})\right) & =\sum_{i=1}^{n}K_{h}\left(\|\bm{s}-\bm{s}_{i}\|\right)\left(y_{i}-\hat{\mu}(\bm{s}_{i};\bm{s})\right)\left\{ V\left(\hat{\mu}(\bm{s}_{i};\bm{s})\right)\right\} ^{-1}\bm{z}_{i}=\bm{0}_{3p},
\end{align}

where $\hat{\mu}(\bm{s}_{i};\bm{s})=g^{-1}\left(\bm{z}_{i}^{T}\hat{\bm{\zeta}}(\bm{s})\right)$
is the mean at location $\bm{s}_{i}$ evaluated at the estimated coefficients
$\hat{\bm{\zeta}}(\bm{s})$ at location $\bm{s}$. The asymptotic
distribution of the local coefficients in a varying-coefficient GLM
with a one-dimensional effect-modifying parameter are given in \citet{Cai-Fan-Li-2000}.
For coefficients that vary in the two dimensions, the arguments in
the proof of Theorem 1 of \citet{Cai-Fan-Li-2000} can be extended
to show that the distribution of the estimated local coefficients
is:

\begin{gather*}
\left\{ nh^{2}f(\bm{s})\right\} ^{1/2}\left[\tilde{\bm{\beta}}(\bm{s})-\bm{\beta}(\bm{s})-(1/2)\kappa_{0}^{-1}\kappa_{2}h^{2}\left\{ \nabla_{uu}^{2}\bm{\beta}(\bm{s})+\nabla_{vv}^{2}\bm{\beta}(\bm{s})\right\} \right]\\
\xrightarrow{{D}}N\left(\bm{0},\kappa_{0}^{-2}\nu_{0}\bm{\Gamma}(\bm{s})^{-1}\right).
\end{gather*}
where $\bm{\Gamma}\left(\bm{s}\right)=E\left\{ \rho\left(\bm{s},\bm{X}(\bm{s})\right)\bm{X}(\bm{s})\bm{X}(\bm{s})^{T}|\bm{s}\right\} $,\\
$\rho(\bm{s},\bm{z})=\left[g_{1}\left(\mu(\bm{s},\bm{z})\right)\right]^{2}Var\left\{ Y(\bm{s})|\bm{X}(\bm{s}),\bm{s}\right\} $,
$g_{1}(\cdot)=g'_{0}(\cdot)/g'(\cdot)$, and $g_{0}(\cdot)$ is the canonical link function.
So when the canonical link is used,
$\rho(\bm{s},\bm{z})=V\left(\mu(\bm{s},\bm{z})\right)$.

\subsection{LAGR Penalized Local Likelihood and Oracle Properties}

Whereas the method of LAGR for local linear regression uses a penalized
local likelihood, LAGR for GLMs uses a penalized negative local quasi-likelihood:

\begin{align*}
\mathcal{J}\left(\bm{\zeta}(\bm{s})\right)&= - \mathcal{\ell}^{*}\left(\bm{\zeta}(\bm{s})\right) + \mathcal{P}\left(\bm{\zeta}(\bm{s})\right)\\
&= - \sum_{i=1}^{n}K_{h}\left(\|\bm{s}-\bm{s}_{i}\|\right)Q\left(g^{-1}\left(\bm{z}_{i}^{T}\bm{\zeta}(\bm{s})\right),Y_{i}\right) + \sum_{j=1}^{p}\phi_{j}(\bm{s})\|\bm{\zeta}_{(j)}(\bm{s})\|.
\end{align*}

Further, let $\phi_{j}(\bm{s})=\lambda_{n}\|\tilde{\bm{\zeta}}_{(j)}(\bm{s})\|^{-\gamma}$,
where $\lambda_{n}>0$ is a the local tuning parameter applied to
all coefficients at location $\bm{s}$ and $\tilde{\bm{\zeta}}_{(j)}(\bm{s})$
is the vector of unpenalized local coefficients. In addition
to the definitions and conditions of Section \ref{sub:oracle-properties},
let  
\[
\bm{\Gamma}_{(a)}\left(\bm{s}\right)=E\left\{ \rho\left(\bm{s},\bm{X}_{(a)}(\bm{s})\right)\bm{X}_{(a)}(\bm{s})\bm{X}_{(a)}(\bm{s})^{T}|\bm{s}\right\}
\]
and assume the following regularity conditions:
\begin{itemize}
\item[(C.8)] The functions $g'''\left(\bm{s}\right)$, $\nabla\bm{\Gamma}\left(\bm{s}\right)$,
$\nabla\bm{\Gamma}_{(a)}\left(\bm{s}\right)$, $V\left(\mu\left(\bm{s},\bm{z}\right)\right)$,
and $V'\left(\mu\left(\bm{s},\bm{z}\right)\right)$ are continuous
at $\bm{s}$.
\item[(C.9)] The function $\left(\partial^{2}/\partial\mu^{2}\right)Q\left(g^{-1}\left(\mu\right),y\right)<0$
for $\mu\in\mathbb{R}$ and $y$ in the range of the response.
\end{itemize}
These additional conditions are not uncommon in the nonparametric
regression literature (see, e.g., conditions (1) and (2) of \citet{Cai-Fan-Li-2000}).
Condition (C.8) is needed for the Taylor's expansion of the local
quasi-likelihood. Condition (C.9) assures that the local quasi-likelihood
is concave and has a unique maximizer.
\begin{thm}[Asymptotic normality]
\label{theorem:normality-glm}  Under (C.1)--(C.10),
\begin{gather*}
\left\{ nh^{2}f(\bm{s})\right\} ^{1/2}\left[\hat{\bm{\beta}}_{(a)}(\bm{s})-\bm{\beta}_{(a)}(\bm{s})-\left(2\kappa_{0}\right)^{-1}\kappa_{2}h^{2}\left\{ \nabla_{uu}^{2}\bm{\beta}_{\left(a\right)}(\bm{s})+\nabla_{vv}^{2}\bm{\beta}_{\left(a\right)}(\bm{s})\right\} \right]\\
\xrightarrow{d}N\left(0,\kappa_{0}^{-2}\nu_{0}\bm{\Gamma}_{(a)}(\bm{s})^{-1}\right)
\end{gather*}
\end{thm}

\begin{thm}[Selection consistency]
\label{theorem:selection-glm}  Under (C.1)--(C.10), if $j>p_{0}(\bm{s})$,
\[
P\left\{ \|\hat{\bm{\zeta}}_{(j)}(\bm{s})\|=\bm{0}\right\} \to1.
\]
 
\end{thm}
By Theorem \ref{theorem:normality-glm}, the LAGR estimates achieve
the same asymptotic distribution as if the nonzero coefficients were
known in advance. The difference between the Gaussian and GLM cases
is that $\sigma^{2}\bm{\Psi}_{(a)}(\bm{s})^{-1}$ in the variance
term of Theorem \ref{theorem:normality} has been replaced by $\bm{\Gamma}_{(a)}(\bm{s})^{-1}$
in Theorem \ref{theorem:normality-glm} because the variance of the
response in the GLM case depends on the expectation of the response.
Theorem \ref{theorem:selection-glm} gives the same result for the
GLM setting as Theorem \ref{theorem:selection} does for the Gaussian
setting: the true zero coefficients are dropped from the model with
probability tending to one. Thus, the oracle properties for the GLM
setting are established. The technical proofs are given in Appendix
B and the necessary lemmas are provided in the online supplementary
materials.

\section{Conclusions and Discussion \label{section-conclusion}}

We have developed a new method of LAGR and shown its oracle properties
for local variable selection and coefficient estimation in VCR models.
This innovation provides a natural and heretofore lacking flexibility to variable selection for
varying coefficient regression models, as any covariate may be included
in part of and not necessarily the entire domain of interest.
This is in contrast to the existing literature on variable selection
for VCR models that focuses on global variable selection, where a covariate is either 
included in or excluded from the model over its entire domain. Further,
the method of LAGR extends the adaptive group Lasso. In particular,
the previous literature on the adaptive group Lasso is insufficient
for local selection in a VCR model because the local weights are functions
of the kernel $K(\cdot)$ and the bandwidth $h$. As a result, the
local observation weights change with sample size and the coefficient
estimates converge at a slower rate than in the traditional adaptive
group Lasso. Under our refined conditions, we have established the oracle property for the LAGR method.

Here we have considered the case of two-dimensional effect-modifying parameter.
Similar results can be obtained when the effect-modifying parameter
has dimension other than two, but in higher dimensions the so-called
``curse of dimensionality'' means that the estimation accuracy may quickly
degrade. Since the optimal rate of convergence for nonparametric
regression is achieved when $h=O\left(n^{-1/\{4+d\}}\right)$ where
$d$ is the dimension of the effect-modifying parameter, it follows
that to attain the oracle properties, the exponent in the adaptive
weights for LAGR estimation must satisfy $\gamma>d/2$.

A possible future direction to take is extension to local regression for spatio-temporal data such that the regression coefficients vary not only in space but also over time. This extension is left to future research.

\section*{References}
\bibliographystyle{chicago}

\clearpage

\appendix
\section{Proofs of Theorems 1--2}
\subsection*{Proof of Theorem \ref{theorem:normality}\label{sec:gaussian-normality-proof} }
\begin{proof}
Let $H_{n}(\bm{u})=\mathcal{J}\left(\bm{\zeta}(\bm{s})+h^{-1}n^{-1/2}\bm{u}\right)-\mathcal{J}\left(\bm{\zeta}(\bm{s})\right)$
and $\alpha_{n}=h^{-1}n^{-1/2}$. Then, we have 
\begin{align*}
H_{n}(\bm{u}) &= (1/2)\left[\bm{Y}-\bm{Z}(\bm{s})\left\{ \bm{\zeta}(\bm{s})+\alpha_{n}\bm{u}\right\} \right]^{T}\bm{W}\!(\bm{s})\left[\bm{Y}-\bm{Z}(\bm{s})\left\{ \bm{\zeta}(\bm{s})+\alpha_{n}\bm{u}\right\} \right]\\
 & +\sum_{j=1}^{p}\phi_{j}(\bm{s})\|\bm{\zeta}_{j}(\bm{s})+\alpha_{n}\bm{u}_{j}\|\\
 & -(1/2)\left\{ \bm{Y}-\bm{Z}(\bm{s})\bm{\zeta}(\bm{s})\right\} ^{T}\bm{W}\!(\bm{s})\left\{ \bm{Y}-\bm{Z}(\bm{s})\bm{\zeta}(\bm{s})\right\} -\sum_{j=1}^{p}\phi_{j}(\bm{s})\|\bm{\zeta}_{j}(\bm{s})\|\\
 &= \left(1/2\right)\alpha_{n}^{2}\bm{u}^{T}\left\{ \bm{Z}(\bm{s})^{T}\bm{W}\!(\bm{s})\bm{Z}(\bm{s})\right\} \bm{u}\\
 & -\alpha_{n}\bm{u}^{T}\left[\bm{Z}(\bm{s})^{T}\bm{W}\!(\bm{s})\left\{ \bm{Y}-\bm{Z}(\bm{s})\bm{\zeta}(\bm{s})\right\} \right]\\
 & +\sum_{j=1}^{p}n^{-1/2}\phi_{j}(\bm{s})n^{1/2}\left\{ \|\bm{\zeta}_{(j)}(\bm{s})+\alpha_{n}\bm{u}_{(j)}\|-\|\bm{\zeta}_{(j)}(\bm{s})\|\right\} .
\end{align*}
The limiting behavior of the last term differs between the cases $j\le p_{0}(\bm{s})$
and $j>p_{0}(\bm{s})$.
\emph{Case $j\le p_{0}(\bm{s})$}: If $j\le p_{0}(\bm{s})$, then $n^{-1/2}\phi_{j}(\bm{s})\to n^{-1/2}\lambda_{n}\|\bm{\zeta}_{(j)}(\bm{s})\|^{-\gamma}$
and \\
$|n^{1/2}\left\{ \|\bm{\zeta}_{(j)}(\bm{s})+\alpha_{n}\bm{u}_{(j)}\|-\|\bm{\zeta}_{(j)}(\bm{s})\|\right\} |\le h^{-1}\|\bm{u}_{(j)}\|$
. Thus, 
\[
\phi_{j}(\bm{s})\left(\|\bm{\zeta}_{(j)}(\bm{s})+\alpha_{n}\bm{u}_{(j)}\|-\|\bm{\zeta}_{(j)}(\bm{s})\|\right)\le\alpha_{n}\phi_{j}(\bm{s})\|\bm{u}_{(j)}\|\le\alpha_{n}a_{n}\|\bm{u}_{(j)}\|\to0.
\]
\emph{Case $j>p_{0}(\bm{s})$}: If $j>p_{0}(\bm{s})$, then $\phi_{j}(\bm{s})\left(\|\bm{\zeta}_{(j)}(\bm{s})+\alpha_{n}\bm{u}_{(j)}\|-\|\bm{\zeta}_{(j)}(\bm{s})\|\right)=\phi_{j}(\bm{s})\alpha_{n}\|\bm{u}_{(j)}\|$.
Since $h=O(n^{-1/6})$, if $hn^{-1/2}b_{n}\xrightarrow{p}\infty$,
then $\alpha_{n}b_{n}\xrightarrow{p}\infty$. Thus, if $\|\bm{u}_{(j)}\|\ne0$,
then 
\[
\alpha_{n}\phi_{j}(\bm{s})\|\bm{u}_{(j)}\|\ge\alpha_{n}b_{n}\|\bm{u}_{(j)}\|\to\infty.
\]
On the other hand, if $\|\bm{u}_{(j)}\|=0$, then $\alpha_{n}\phi_{j}(\bm{s})\|\bm{u}_{(j)}\|=0$.
Thus, the limit of $H_{n}(\bm{u})$ is the same as the limit of $H_{n}^{*}(\bm{u})$
where $H_{n}^{*}(\bm{u})=\infty$ if $\|\bm{u}_{(j)}\|\ne0$ for some
$j>p_{0}(\bm{s})$, and 
\[
H_{n}^{*}(\bm{u})=(1/2)\alpha_{n}^{2}\bm{u}^{T}\left\{ \bm{Z}(\bm{s})^{T}\bm{W}\!(\bm{s})\bm{Z}(\bm{s})\right\} \bm{u}-\alpha_{n}\bm{u}^{T}\left[\bm{Z}(\bm{s})^{T}\bm{W}\!(\bm{s})\left\{ \bm{Y}-\bm{Z}(\bm{s})\bm{\zeta}(\bm{s})\right\} \right]
\]
otherwise. It follows that $H_{n}^{*}(\bm{u})$ is convex and has
a unique minimizer, called $\hat{\bm{u}}_{n}$. Let $\hat{\bm{u}}_{(a)n}$
and $\hat{\bm{u}}_{(b)n}$ be, respectively, the subvectors of $\bm{u}_{n}$
corresponding to the true nonzero coefficients and true zero coefficients.
Then 
\[
\hat{\bm{u}}_{(a)n}=\left\{ n^{-1}\bm{Z}_{(a)}(\bm{s})^{T}\bm{W}\!(\bm{s})\bm{Z}_{(a)}(\bm{s})\right\} ^{-1}\left[hn^{1/2}\bm{Z}_{(a)}(\bm{s})^{T}\bm{W}\!(\bm{s})\left\{ \bm{Y}-\bm{Z}_{(a)}(\bm{s})\bm{\zeta}_{(a)}(\bm{s})\right\} \right]
\]
and $\hat{\bm{u}}_{(b)n}=\bm{0}.$ By epiconvergence, the minimizer
of the limiting function is the limit of the minimizers $\hat{\bm{u}}_{n}$
\citep{Geyer-1994,Knight-Fu-2000}. Since, by Lemma 2 of \citet{Sun-Yan-Zhang-Lu-2014},
\[
\hat{\bm{u}}_{(a)n}-\left(2\alpha_{n}f(\bm{s})^{1/2}\kappa_{0}\right)^{-1}\kappa_{2}h^{2}\left\{ \nabla_{uu}^{2}\bm{\zeta}_{(a)}(\bm{s})+\nabla_{vv}^{2}\bm{\zeta}_{(a)}(\bm{s})\right\} \xrightarrow{d}N\left(0,\alpha_{n}^{-2}f(\bm{s})^{-1}\kappa_{0}^{-2}\nu_{0}\sigma^{2}\bm{\Psi}_{(a)}(\bm{s})^{-1}\right)
\]
the result of Theorem \ref{theorem:normality} follows.
\end{proof}
\subsection*{Proof of Theorem \ref{theorem:selection}\label{sec:gaussian-selection-proof}}
\begin{proof}
The proof is by contradiction. Without loss of generality we consider
only the $p$th covariate group. Assume $\|\hat{\bm{\zeta}}_{(p)}(\bm{s})\|\ne0$.
Then $\mathcal{J}\left(\bm{\zeta}(\bm{s})\right)$ is differentiable
w.r.t. $\bm{\zeta}_{(p)}(\bm{s})$ and is minimized where 
\begin{align*}
\bm{0} &= \bm{Z}_{(p)}^{T}(\bm{s})\bm{W}\!(\bm{s})\left\{ \bm{Y}-\bm{Z}_{(\mhyphen p)}\left(\bm{s}\right)\hat{\bm{\zeta}}_{(\mhyphen p)}(\bm{s})-\bm{Z}_{(p)}\left(\bm{s}\right)\hat{\bm{\zeta}}_{(p)}(\bm{s})\right\} -\phi_{(p)}(\bm{s})\hat{\bm{\zeta}}_{(p)}(\bm{s})\|\hat{\bm{\zeta}}_{(p)}(\bm{s})\|^{-1}\\
 &= \bm{Z}_{(p)}(\bm{s})^{T}\bm{W}\!(\bm{s})\left[\bm{Y}-\bm{Z}(\bm{s})\bm{\zeta}(\bm{s})-\left(2\kappa_{0}\right)^{-1}h^{2}\kappa_{2}\left\{ \nabla_{uu}^{2}\bm{\zeta}(\bm{s})+\nabla_{vv}^{2}\bm{\zeta}(\bm{s})\right\} \right]\\
 & +\bm{Z}_{(p)}(\bm{s})^{T}\bm{W}\!(\bm{s})\bm{Z}_{(\mhyphen p)}(\bm{s})\left[\bm{\zeta}_{(\mhyphen p)}(\bm{s})+\left(2\kappa_{0}\right)^{-1}h^{2}\kappa_{2}\left\{ \nabla_{uu}^{2}\bm{\zeta}_{(\mhyphen p)}(\bm{s})+\nabla_{vv}^{2}\bm{\zeta}_{(\mhyphen p)}(\bm{s})\right\} -\hat{\bm{\zeta}}_{(\mhyphen p)}(\bm{s})\right]\\
 & +\bm{Z}_{(p)}(\bm{s})^{T}\bm{W}\!(\bm{s})\bm{Z}_{(p)}(\bm{s})\left[\bm{\zeta}_{(p)}(\bm{s})+\left(2\kappa_{0}\right)^{-1}h^{2}\kappa_{2}\left\{ \nabla_{uu}^{2}\bm{\zeta}_{(p)}(\bm{s})+\nabla_{vv}^{2}\bm{\zeta}_{(p)}(\bm{s})\right\} -\hat{\bm{\zeta}}_{(p)}(\bm{s})\right]\\
 & -\phi_{p}(\bm{s})\hat{\bm{\zeta}}_{(p)}(\bm{s})\|\hat{\bm{\zeta}}_{(p)}(\bm{s})\|^{-1}.
\end{align*}
Thus,
\begin{align}
&\left(n^{-1}h^{2}\right)^{1/2} \phi_{p}(\bm{s})\hat{\bm{\zeta}}_{(p)}(\bm{s})\|\hat{\bm{\zeta}}_{(p)}(\bm{s})\|^{-1}\nonumber \\
 =& \bm{Z}_{(p)}(\bm{s})^{T}\bm{W}\!(\bm{s})\left(n^{-1}h^{2}\right)^{1/2}\left[\bm{Y}-\bm{Z}(\bm{s})\bm{\zeta}(\bm{s})-\frac{h^{2}\kappa_{2}}{2\kappa_{0}}\left\{ \nabla_{uu}^{2}\bm{\zeta}(\bm{s})+\nabla_{vv}^{2}\bm{\zeta}(\bm{s})\right\} \right]\nonumber \\
 & \mkern-18mu+\left\{ n^{-1}\bm{Z}_{(p)}(\bm{s})^{T}\bm{W}\!(\bm{s})\bm{Z}_{(\mhyphen p)}(\bm{s})\right\}\!\! \left(nh^{2}\right)^{1/2} \! \left[\bm{\zeta}_{(\mhyphen p)}(\bm{s}) \!+\! \frac{h^{2}\kappa_{2}}{2\kappa_{0}} \!\left\{ \nabla_{uu}^{2}\bm{\zeta}_{(\mhyphen p)}(\bm{s}) \!+\! \nabla_{vv}^{2}\bm{\zeta}_{(\mhyphen p)}(\bm{s})\right\} \!-\! \hat{\bm{\zeta}}_{(\mhyphen p)}(\bm{s})\right]\nonumber \\
 & \mkern-18mu+\left\{ n^{-1}\bm{Z}_{(p)}(\bm{s})^{T}\bm{W}\!(\bm{s})\bm{Z}_{(p)}(\bm{s})\right\}\!\! \left(nh^{2}\right)^{1/2} \! \left[\bm{\zeta}_{(p)}(\bm{s}) \!+\! \frac{h^{2}\kappa_{2}}{2\kappa_{0}} \!\left\{ \nabla_{uu}^{2}\bm{\zeta}_{(p)}(\bm{s}) \!+\! \nabla_{vv}^{2}\bm{\zeta}_{(p)}(\bm{s})\right\} \!-\! \hat{\bm{\zeta}}_{(p)}(\bm{s})\right].\label{eq:selection}
\end{align}
From Lemma 2 of \citet{Sun-Yan-Zhang-Lu-2014}, 
\[
O_{p}\left(n^{-1}\bm{Z}_{(p)}(\bm{s})^{T}\bm{W}\!\left(\bm{s}\right)\bm{Z}_{(\mhyphen p)}(\bm{s})\right)=O_{p}\left(n^{-1}\bm{Z}_{(p)}(\bm{s})^{T}\bm{W}\!(\bm{s})\bm{Z}_{(p)}(\bm{s})\right)=O_{p}\left(1\right).
\]
From Theorem 3 of \citet{Sun-Yan-Zhang-Lu-2014}, we have that 
\[
\left(nh^{2}\right)^{1/2}\left[\hat{\bm{\zeta}}_{(\mhyphen p)}(\bm{s})-\bm{\zeta}_{(\mhyphen p)}(\bm{s})-\left(2\kappa_{0}\right)^{-1}h^{2}\kappa_{2}\left\{ \nabla_{uu}^{2}\bm{\zeta}_{(\mhyphen p)}(\bm{s})+\nabla_{vv}^{2}\bm{\zeta}_{(\mhyphen p)}(\bm{s})\right\} \right]=O_{p}\left(1\right)
\]
 and 
\[
\left(nh^{2}\right)^{1/2}\left[\hat{\bm{\zeta}}_{(p)}(\bm{s})-\bm{\zeta}_{(p)}(\bm{s})-\left(2\kappa_{0}\right)^{-1}h^{2}\kappa_{2}\left\{ \nabla_{uu}^{2}\bm{\zeta}_{(p)}(\bm{s})+\nabla_{vv}^{2}\bm{\zeta}_{(p)}(\bm{s})\right\} \right]=O_{p}\left(1\right).
\]
We showed in the proof of Theorem \ref{theorem:normality} that
\[
\left(nh^{2}\right)^{1/2}\bm{Z}_{(p)}(\bm{s})^{T}\bm{W}\!(\bm{s})\left[\bm{Y}-\bm{Z}\!(\bm{s})\bm{\zeta}(\bm{s})-\left(2\kappa_{0}\right)^{-1}h^{2}\kappa_{2}\left\{ \nabla_{uu}^{2}\bm{\zeta}(\bm{s})+\nabla_{vv}^{2}\bm{\zeta}(\bm{s})\right\} \right]=O_{p}\left(1\right).
\]
The right hand side of (\ref{eq:selection}) is $O_{p}(1)$, so for
$\hat{\bm{\zeta}}_{(p)}(\bm{s})$ to be a solution, we must have that
$hn^{-1/2}\phi_{p}(\bm{s})\hat{\bm{\zeta}}_{(p)}(\bm{s})\|\hat{\bm{\zeta}}_{(p)}(\bm{s})\|^{-1}=O_{p}\left(1\right)$.
But since by assumption $\hat{\bm{\zeta}}_{(p)}(\bm{s})\ne\bm{0}$,
there must be some $k\in\{1,2,3\}$ such that $|\hat{\zeta}_{(p)_{k}}(\bm{s})|=\max\{|\hat{\zeta}_{(p)_{m}}(\bm{s})|:1\le m\le3\}$.
And for this $k$, we have that $|\hat{\zeta}_{(p)_{k}}(\bm{s})|\|\hat{\bm{\zeta}}_{(p)}(\bm{s})\|^{-1}\ge3^{-1/2}>0$.
Since $hn^{-1/2}b_{n}\to\infty$, we have that $hn^{-1/2}\phi_{p}(\bm{s})\hat{\bm{\zeta}}_{(p)}(\bm{s})\|\hat{\bm{\zeta}}_{(p)}(\bm{s})\|^{-1}\ge hb_{n}\left(3n\right)^{-1/2}\to\infty$
and therefore the left hand side of (\ref{eq:selection}) dominates
the sum to the right side. Thus, for large enough $n$, $\hat{\bm{\zeta}}_{(p)}(\bm{s})\ne\bm{0}$
cannot maximize $\mathcal{J}\left(\cdot\right)$, and therefore $P\left\{ \hat{\bm{\zeta}}_{(b)}(\bm{s})=\bm{0}\right\} \to1$. 
\end{proof}
\section{Proofs of Theorems 3--4}
\subsection*{Proof of Theorem \ref{theorem:normality-glm}}
The next proofs require the lemmas in the web-based supplemental material.
First, let $\bm{z}\in\mathbb{R}^{3p}$. Define the $q$-functions
to be the derivatives of the quasi-likelihood: $q_{j}(t,y)=\left(\partial/\partial t\right)^{j}Q\left(g^{-1}(t),y\right)$.
Then $q_{1}\left(\eta\left(\bm{s},\bm{z}\right),\mu\left(\bm{s},\bm{z}\right)\right)=\bm{0}$,
and $q_{2}\left(\eta\left(\bm{s},\bm{z}\right),\mu\left(\bm{s},\bm{z}\right)\right)=-\rho\left(\bm{s},\bm{z}\right)$.
Let 
\[
\tilde{\bm{\beta}}_{i}''=\left[\left(\bm{s}_{i}-\bm{s}\right)^{T}\left\{ \nabla^{2}\beta_{1}(\bm{s})\right\} \left(\bm{s}_{i}-\bm{s}\right),\dots,\left(\bm{s}_{i}-\bm{s}\right)^{T}\left\{ \nabla^{2}\beta_{p}(\bm{s})\right\} \left(\bm{s}_{i}-\bm{s}\right)\right]^{T}
\]
 be the $p$-vector of quadratic forms of location interactions on
the second derivatives of the coefficient functions.
\begin{proof}
Let $H'_{n}(\bm{u})=\mathcal{J}^{*}\left(\bm{\zeta}(\bm{s})+\alpha_{n}\bm{u}\right)-\mathcal{J}^{*}\left(\bm{\zeta}(\bm{s})\right)$
and $\alpha_{n}=h^{-1}n^{-1/2}$. Then, minimizing $H'_{n}(\bm{u})$
is equivalent to minimizing $H_{n}(\bm{u})$, where 
\begin{align*}
H_{n}(\bm{u}) &= - n^{-1}\sum_{i=1}^{n}Q\left(g^{-1}\left(\bm{Z}_{i}^{T}\left\{ \bm{\zeta}(\bm{s})+\alpha_{n}\bm{u}\right\} \right),Y_{i}\right)K\left(h^{-1}\|\bm{s}-\bm{s}_{i}\|\right)\\
 & + n^{-1}\sum_{i=1}^{n}Q\left(g^{-1}\left(\bm{Z}_{i}^{T}\bm{\zeta}(\bm{s})\right),Y_{i}\right)K\left(h^{-1}\|\bm{s}-\bm{s}_{i}\|\right)\\
 & +n^{-1}\sum_{j=1}^{p}\phi_{j}\left(\bm{s}\right)\|\bm{\zeta}_{(j)}(\bm{s})+\alpha_{n}\bm{u}\|-\sum_{j=1}^{p}\phi_{j}\left(\bm{s}\right)\|\bm{\zeta}_{(j)}(\bm{s})\|.
\end{align*}
Define
\[
\Omega_{n}=\alpha_{n}\sum_{i=1}^{n}q_{1}\left(\bm{Z}_{i}^{T}\bm{\zeta}(\bm{s}),Y_{i}\right)\bm{Z}_{i}K\left(h^{-1}\|\bm{s}-\bm{s}_{i}\|\right)=\alpha_{n}\sum_{i=1}^{n}\omega_{i}
\]
and 
\[
\Delta_{n}= - \alpha_{n}^{2}\sum_{i=1}^{n}q_{2}\left(\bm{Z}_{i}^{T}\bm{\zeta}(\bm{s}),Y_{i}\right)\bm{Z}_{i}\bm{Z}_{i}^{T}K\left(h^{-1}\|\bm{s}-\bm{s}_{i}\|\right)=\alpha_{n}^{2}\sum_{i=1}^{n}\delta_{i}.
\]
Then it follows from the Taylor expansion of $\mathcal{J}^{*}\left(\bm{\zeta}(\bm{s})+\alpha_{n}\bm{u}\right)$
around $\bm{\zeta}(\bm{s})$ that
\begin{align}
H_{n}\left(\bm{u}\right) &= -\Omega_{n}^{T}\bm{u}+(1/2)\bm{u}^{T}\Delta_{n}\bm{u}+\left(\alpha_{n}^{3}/6\right)\sum_{i=1}^{n}q_{3}\left(\bm{Z}_{i}^{T}\tilde{\bm{\zeta}}_{i},Y_{i}\right)\left[\bm{Z}_{i}^{T}\bm{u}\right]^{3}K\left(h^{-1}\|\bm{s}-\bm{s}_{i}\|\right)\nonumber \\
 & +\sum_{j=1}^{p}\phi_{j}\left(\bm{s}\right)\left\{ \|\bm{\zeta}_{(j)}(\bm{s})+h^{-1}n^{-1/2}\bm{u}\|-\|\bm{\zeta}_{(j)}(\bm{s})\|\right\} .\label{eq:taylor-expanded-glm-criterion}
\end{align}
where $\tilde{\bm{\zeta}_{i}}$ lies between $\bm{\zeta}(\bm{s})$
and $\bm{\zeta}(\bm{s})+\alpha_{n}\bm{u}$. Since $q_{3}\left(\bm{Z}_{i}^{T}\tilde{\bm{\zeta}}_{i},Y_{i}\right)$
is linear in $Y_{i}$, $K\left(\cdot\right)$ is bounded, and, by
condition (C.6),
\[
\left(\alpha_{n}^{3}/6\right)E\left|\sum_{i=1}^{n}q_{3}\left(\bm{Z}_{i}^{T}\tilde{\bm{\zeta}}_{i},Y_{i}\right)\left[\bm{Z}_{i}^{T}\bm{u}\right]^{3}K\left(h^{-1}\|\bm{s}-\bm{s}_{i}\|\right)\right|=O\left(\alpha_{n}\right),
\]
the third term in (\ref{eq:taylor-expanded-glm-criterion}) is $O_{p}\left(\alpha_{n}\right)$.
The limiting behavior of the last term of (\ref{eq:taylor-expanded-glm-criterion})
differs between the cases $j\le p_{0}(\bm{s})$ and $j>p_{0}(\bm{s})$.
\emph{Case $j\le p_{0}(\bm{s})$:} If $j\le p_{0}(\bm{s})$, then $n^{-1/2}\phi_{j}(\bm{s})\to n^{-1/2}\lambda_{n}\|\bm{\zeta}_{(j)}(\bm{s})\|^{-\gamma}$
and $|\sqrt{n}\left\{ \|\bm{\zeta}_{(j)}(\bm{s})+\alpha_{n}\bm{u}_{(j)}\|-\|\bm{\zeta}_{(j)}(\bm{s})\|\right\} |\le h^{-1}\|\bm{u}_{(j)}\|$. Thus, 
\[
\lim\limits _{n\to\infty}\phi_{j}(\bm{s})\left(\|\bm{\zeta}_{(j)}(\bm{s})+\alpha_{n}\bm{u}_{(j)}\|-\|\bm{\zeta}_{(j)}(\bm{s})\|\right)\le\alpha_{n}\phi_{j}(\bm{s})\|\bm{u}_{(j)}\|\le\alpha_{n}a_{n}\|\bm{u}_{(j)}\|\to0
\]
\emph{Case $j>p_{0}(\bm{s})$:} If $j>p_{0}(\bm{s})$, then $\phi_{j}(\bm{s})\left(\|\bm{\zeta}_{(j)}(\bm{s})+\alpha_{n}\bm{u}_{(j)}\|-\|\bm{\zeta}_{(j)}(\bm{s})\|\right)=\phi_{j}(\bm{s})\alpha_{n}\|\bm{u}_{(j)}\|$.
Since $h=O(n^{-1/6})$, if $hn^{-1/2}b_{n}\xrightarrow{p}\infty$,
then $\alpha_{n}b_{n}\xrightarrow{p}\infty$. Now, if $\|\bm{u}_{(j)}\|\ne0$,
then 
\[
\alpha_{n}\phi_{j}(\bm{s})\|\bm{u}_{(j)}\|\ge\alpha_{n}b_{n}\|\bm{u}_{(j)}\|\to\infty.
\]
On the other hand, if $\|\bm{u}_{(j)}\|=0$, then $\alpha_{n}\phi_{j}(\bm{s})\|\bm{u}_{(j)}\|=0$.
By Lemma 1, $\Delta_{n}=\Delta+O_{p}\left(\alpha_{n}\right)$,
so the limit of $H_{n}(\bm{u})$ is the same as the limit of $H_{n}^{*}(\bm{u})$
where
\[
H_{n}^{*}(\bm{u})= -\Omega_{(a)n}^{T}\bm{u}_{(a)}+(1/2)\bm{u}_{(a)}^{T}\Delta_{(a)}\bm{u}_{(a)}+o_{p}\left(1\right)
\]
if $\|\bm{u}_{j}\|=0\;\forall j>p_{0}(\bm{s})$, and $H_{n}^{*}(\bm{u})=\infty$
otherwise. It follows that $H_{n}^{*}(\bm{u})$ is convex and has
a unique minimizer, called $\hat{\bm{u}}_{n}$. Let $\hat{\bm{u}}_{(a)n}$
$\Delta_{(a)}$ and $\Omega_{(a)n}$ be, respectively, the parts of
$\bm{u}_{n}$, $\Delta$, and $\Omega_{n}$ corresponding to the true
nonzero coefficients, and let $\hat{\bm{u}}_{(b)n}$ be the subvector
of $\hat{\bm{u}}_{n}$ corresponding to the true zero coefficients.
Then
\begin{align*}
\hat{\bm{u}}_{(a)n} &= \Delta_{(a)}^{-1}\Omega_{(a)n}+o_{p}\left(1\right)\text{ and }\hat{\bm{u}}_{(b)n}=\bm{0}
\end{align*}
by the quadratic approximation lemma \citep{Fan-Gijbels-1996}. By epiconvergence, the minimizer of the limiting function is the limit
of the minimizers $\hat{\bm{u}}_{n}$ \citep{Geyer-1994,Knight-Fu-2000}.
Since $\Delta$ is a constant, the normality of $\hat{\bm{u}}_{(a)n}$
follows from the normality of $\Omega_{n}$, which is established
via the Cram\'{e}r-Wold device. Let $\bm{d}\in\mathbb{R}^{3p}$ be
a unit vector, and let
\[
\xi_{i}=q_{1}\left(\bm{Z}_{i}^{T}\bm{\zeta}(\bm{s}),Y_{i}\right)\bm{d}^{T}\bm{Z}_{i}K\left(h^{-1}\|\bm{s}_{i}-\bm{s}\|\right).
\]
Then $\bm{d}^{T}\Omega_{n}=\alpha_{n}\sum_{i=1}^{n}\xi_{i}$. We establish
the normality of $\bm{d}^{T}\Omega_{n}$ by checking the Lyapunov
condition of the sequence $\left\{ \bm{d}^{T}Var\left(\Omega_{n}\right)\bm{d}\right\} ^{-1/2}\left\{ \bm{d}^{T}\Omega_{n}-\bm{d}^{T}E\Omega_{n}\right\} $.
By boundedness of $K\left(\cdot\right)$, linearity of $q_{1}\left(\bm{Z}_{i}^{T}\bm{\zeta}(\bm{s}),Y_{i}\right)$
in $Y_{i}$, and conditions (C.6) and (C.8), we have that
\begin{equation}
n\alpha_{n}^{3}E\left(\left|\xi_{1}\right|^{3}\right)=O\left(\alpha_{n}\right)\to0.\label{eq:lyapunov-bound}
\end{equation}
We observe that (\ref{eq:lyapunov-bound}) implies that $n\alpha_{n}^{3}\left|E\left(\xi_{1}\right)\right|^{3}\to0$,
and since $E\left(\left|\xi_{1}-E\xi_{1}\right|^{3}\right)<E\left\{ \left(\left|\xi_{1}\right|+\left|E\xi_{1}\right|\right)^{3}\right\} \to0$,
the Lyapunov condition is satisfied. Thus, $\Omega_{n}$ asymptotically
follows a Gaussian distribution and the result follows from the quadratic
approximation lemma.
\end{proof}
\subsection*{Proof of Theorem \ref{theorem:selection-glm}}
\begin{proof}
The proof is by contradiction. Without loss of generality we consider
only the $p$th covariate group.
Assume $\|\hat{\bm{\zeta}}_{(p)}(\bm{s})\|\ne0$. Then $\mathcal{J}\left(\bm{\zeta}(\bm{s})\right)$
is differentiable w.r.t. $\bm{\zeta}_{(p)}(\bm{s})$ and is minimized
where 
\begin{align}
\phi_{p}(\bm{s})\hat{\bm{\zeta}}_{(p)}(\bm{s})\|\hat{\bm{\zeta}}_{(p)}(\bm{s})\|^{-1} &= \sum_{i=1}^{n}q_{1}\!\left(\bm{Z}_{i}^{T}\hat{\bm{\zeta}}(\bm{s}),Y_{i}\right)\bm{Z}_{i(p)}K\left(h^{-1}\|\bm{s}_{i}-\bm{s}\|\right)\label{eq:glm-selection}
\end{align}
From Lemma 2, the right hand side of (\ref{eq:glm-selection})
is $O_{p}\left(1\right)$, so for $\hat{\bm{\zeta}}_{(p)}(\bm{s})$
to be a solution, we must have that $hn^{-1/2}\phi_{p}(\bm{s})\hat{\bm{\zeta}}_{(p)}(\bm{s})\|\hat{\bm{\zeta}}_{(p)}(\bm{s})\|^{-1}=O_{p}\left(1\right)$.
But since by assumption $\hat{\bm{\zeta}}_{(p)}(\bm{s})\ne\bm{0}$,
there must be some $k\in\{1,2,3\}$ such that $|\hat{\zeta}_{(p)_{k}}(\bm{s})|=\max\{|\hat{\zeta}_{(p)_{m}}(\bm{s})|:1\le m\le3\}$.
And for this $k$, we have that $|\hat{\zeta}_{(p)_{k}}(\bm{s})|\|\hat{\bm{\zeta}}_{(p)}(\bm{s})\|^{-1}\ge3^{-1/2}>0$.
Since $hn^{-1/2}b_{n}\to\infty$, we have that $hn^{-1/2}\phi_{p}(\bm{s})\hat{\bm{\zeta}}_{(p)}(\bm{s})\|\hat{\bm{\zeta}}_{(p)}(\bm{s})\|^{-1}\ge hb_{n}(3n)^{-1/2}\to\infty$
and therefore the left hand side of (\ref{eq:glm-selection}) dominates
the sum to the right side. Thus, for large enough $n$, $\hat{\bm{\zeta}}_{(p)}(\bm{s})\ne\bm{0}$
cannot maximize $\mathcal{J}\left(\cdot\right)$, and therefore $P\left\{ \hat{\bm{\zeta}}_{(b)}(\bm{s})=\bm{0}\right\} \to1$. 
\end{proof}

\end{document}


\begin{frontmatter}

\title{Web-based Supplementary Material for ``Local Adaptive Grouped Regularization and its Oracle Properties for Varying Coefficient Regression"}

\author[wrbrooks]{Wesley Brooks}
\ead{wrbrooks@uwalumni.com}

\author[jzhu]{Jun Zhu}
\ead{jzhu@stat.wisc.edu}

\author[zlu]{Zudi Lu}
\ead{Z.Lu@soton.ac.uk}

\address[wrbrooks]{Department of Statistics, University of Wisconsin, Madison, WI 53706}
\address[jzhu]{Department of Statistics and Department of Entomology, University of Wisconsin, Madison, WI 53706}
\address[zlu]{School of Mathematical Sciences, The University of Southampton, Highfield, Southampton UK}

\begin{abstract}
The lemmas in this supplement are used in the proofs of Theorems 3 and 4 of the main text.
\end{abstract}

\begin{keyword}
\end{keyword}

\end{frontmatter}

\section{Lemmas}
\begin{lem}
\label{lemma:omega}
\begin{multline*}
E\left[\sum_{i=1}^{n}q_{1}\left(\bm{Z}_{i}^{T}\bm{\zeta}(\bm{s}),Y_{i}\right)\bm{Z}_{i}K_{h}\left(\|\bm{s}-\bm{s}_{i}\|\right)\right]=\\
\left(\begin{array}{c}
2^{-1}n^{1/2}h^{3}f(\bm{s})\kappa_{2}\bm{\Gamma}(\bm{s})\left(\nabla_{uu}^{2}\bm{\beta}(\bm{s})+\nabla_{vv}^{2}\bm{\beta}(\bm{s})\right)^{T}\\
\bm{0}_{2p}
\end{array}\right)+o_{p}\left(h^{2}\bm{1}_{3p}\right)
\end{multline*}
and
\begin{align*}
Var\left[\sum_{i=1}^{n}q_{1}\left(\bm{Z}_{i}^{T}\bm{\zeta}(\bm{s}),Y_{i}\right)\bm{Z}_{i}K_{h}\left(\|\bm{s}-\bm{s}_{i}\|\right)\right] &= f(\bm{s}){\rm diag}\left\{ \nu_{0},\nu_{2},\nu_{2}\right\} \otimes\bm{\Gamma}(\bm{s})+o\left(1\right)\\
&= \Lambda+o\left(1\right)
\end{align*}
\end{lem}
\begin{proof}
\textbf{Expectation}: For $j=1,\dots,p$, by a Taylor expansion of
$\beta_{j}(\bm{s}_{i})$ around $\bm{s}$,
\[
\beta_{j}(\bm{s}_{i})=\beta_{j}(\bm{s})+\nabla\beta_{j}(\bm{s})(\bm{s}_{i}-\bm{s})+(\bm{s}_{i}-\bm{s})^{T}\left\{ \nabla^{2}\beta_{j}(\bm{s})\right\} (\bm{s}_{i}-\bm{s})+o\left(h^{2}\right)
\]
and thus, for $\bm{x}\in\mathbb{R}^{p}$, 
\[
\bm{x}_{i}^{T}\bm{\beta}\!\left(\bm{s}_{i}\right)=\sum_{j=1}^{p}x_{ij}\left[\beta_{j}(\bm{s})+\nabla\beta_{j}(\bm{s})^{T}(\bm{s}_{i}-\bm{s})+\tilde{\beta}''_{ij}\right]+o\left(h^{2}\right).
\]
Letting $\bm{z}_{i}^{T}=\left\{ \left(1,s_{i,1}-s_{1},s_{i,2}-s_{2}\right)\otimes\bm{x}_{i}^{T}\right\} $
and $\bm{\zeta}(\bm{s})=\left(\bm{\beta}(\bm{s})^{T},\nabla_{u}\bm{\beta}(\bm{s})^{T},\nabla_{v}\bm{\beta}(\bm{s})^{T}\right)^{T}$,
we have that 
\begin{align*}
\bm{x}_{i}^{T}\bm{\beta}(\bm{s}_{i})-\bm{z}_{i}^{T}\bm{\zeta}(\bm{s}) &= \bm{x}_{i}^{T}\tilde{\bm{\beta}}''_{i}+o\left(h^{2}\right)\\
&= O_{p}\left(h^{2}\right).
\end{align*}
By a Taylor expansion around $\bm{x}^{T}\bm{\beta}(\bm{s}_{i})$,
then, 
\begin{align*}
q_{1}\left(\bm{z}_{i}^{T}\bm{\zeta}(\bm{s}),\mu(\bm{s}_{i},\bm{z}_{i})\right) &= q_{1}\left(\bm{x}_{i}^{T}\bm{\beta}(\bm{s}_{i}),\mu(\bm{s}_{i},\bm{z})\right)\\
 & -q_{2}\left(\bm{x}_{i}^{T}\bm{\beta}(\bm{s}_{i}),\mu(\bm{s}_{i},\bm{z})\right)\bm{x}_{i}^{T}\tilde{\bm{\beta}}''_{i}\\
 & +o\left(h^{2}\right).
\end{align*}
And by the definitions of $q_{1}(\cdot)$ and $q_{2}(\cdot)$, we
have that
\[
q_{1}\left(\bm{z}_{i}^{T}\bm{\zeta}(\bm{s}),\mu(\bm{s}_{i},\bm{z}_{i})\right)=\rho(\bm{s}_{i},\bm{z}_{i})\bm{x}_{i}^{T}\tilde{\bm{\beta}}''_{i}+o\left(h^{2}\right).
\]
Now the expectation of $\Omega_{n}$ is 
\begin{align*}
nE\left(\omega_{i}|\bm{Z}_{i}=\bm{z}_{i},\bm{s}_{i}\right) &= \left(1/2\right)\alpha_{n}\bm{z}_{i}q_{1}\left(\bm{z}_{i}^{T}\bm{\zeta}(\bm{s}),\mu(\bm{s}_{i},\bm{z}_{i})\right)K\left(h^{-1}\|\bm{s}-\bm{s}_{i}\|\right)\\
&= \left(1/2\right)\alpha_{n}h^{2}\bm{z}_{i}\left\{ h^{-2}\rho(\bm{s}_{i},\bm{z}_{i})\bm{x}_{i}^{T}\tilde{\bm{\beta}}''_{i}+o\left(\bm{1}_{3p}\right)\right\} K\left(h^{-1}\|\bm{s}-\bm{s}_{i}\|\right).
\end{align*}
To facilitate a change of variables, we observe that $h^{-2}\tilde{\beta}''_{j}=\left(\frac{\bm{s}_{i}-\bm{s}}{h}\right)^{T}\left\{ \nabla^{2}\beta_{j}(\bm{s})\right\} \left(\frac{\bm{s}_{i}-\bm{s}}{h}\right)$.
Thus,
\begin{align*}
E\left(\omega_{i}|\bm{s}_{i}\right) &= \left(1/2\right)\alpha_{n}h^{2}\left[\left(\begin{array}{c}
1\\
h^{-1}(s_{i,1}-s_{1})\\
h^{-1}(s_{i,2}-s_{2})
\end{array}\right)\otimes\left\{ \bm{\Gamma}(\bm{s}_{i})h^{-2}\tilde{\bm{\beta}}''_{i}\right\} +o\left(\bm{1}_{3p}\right)\right]K\left(h^{-1}\|\bm{s}-\bm{s}_{i}\|\right).
\end{align*}
And, using the symmetry of the kernel function,
\begin{align*}
E\left(\omega_{i}\right) &= (1/2)\alpha_{n}h^{4}f(\bm{s})\left(\begin{array}{c}
\kappa_{2}\\
h\kappa_{3}\\
h\kappa_{3}
\end{array}\right)\otimes\left[\bm{\Gamma}(\bm{s})\left\{ \nabla_{uu}^{2}\bm{\beta}(\bm{s})+\nabla_{vv}^{2}\bm{\beta}(\bm{s})\right\} \right]+o\left(h^{2}\bm{1}_{3p}\right)
\end{align*}
where $\left\{ \nabla_{uu}^{2}\bm{\beta}(\bm{s})+\nabla_{vv}^{2}\bm{\beta}(\bm{s})\right\} =\left(\nabla_{uu}^{2}\beta_{1}(\bm{s})+\nabla_{vv}^{2}\beta_{1}(\bm{s}),\dots,\nabla_{uu}^{2}\beta_{p}(\bm{s})+\nabla_{vv}^{2}\beta_{p}(\bm{s})\right)^{T}$.
Thus,
\begin{align*}
E\left(\Omega_{n}\right) &= \left(\begin{array}{c}
\alpha_{n}^{-1}2^{-1}h^{2}\kappa_{2}f(\bm{s})\bm{\Gamma}(\bm{s})\left(\nabla_{uu}^{2}\bm{\beta}(\bm{s})+\nabla_{vv}^{2}\bm{\beta}(\bm{s})\right)^{T}\\
\bm{0}_{2p}
\end{array}\right)+o_{p}\left(h^{2}\bm{1}_{3p}\right)
\end{align*}
\textbf{Variance}: By the previous result, $E\left(\Omega_{n}\right)=O\left(h^{2}\right)$.
Thus, $var\left(\Omega_{n}\right)\to E\left(\Omega_{n}^{2}\right)$,
and since the observations are independent, $E\left(\Omega_{n}^{2}\right)=\sum_{i=1}^{n}E\left(\omega_{i}^{2}\right)$.
And, by Taylor expansion around $\bm{z}_{i}^{T}\bm{\zeta}(\bm{s}_{i})$, 
\begin{align*}
q_{1}^{2}\left(\bm{z}_{i}^{T}\bm{\zeta}(\bm{s}),Y_{i}\right) &=  q_{1}^{2}\left(\bm{z}_{i}^{T}\bm{\zeta}(\bm{s}_{i}),Y_{i}\right)\\
 & -q_{1}\left(\bm{z}_{i}^{T}\bm{\zeta}(\bm{s}_{i}),Y_{i}\right)q_{2}\left(\bm{z}_{i}^{T}\bm{\zeta}(\bm{s}_{i}),Y_{i}\right)\bm{x}_{i}^{T}\tilde{\bm{\beta}}''_{i}\\
 & +o\left(h^{2}\right).
\end{align*}
Since $q_{1}\left(\cdot,\cdot\right)$ is the quasi-score function,
it follows that 
\begin{align*}
E\left(\omega_{i}^{2}|\bm{Z}_{i}=\bm{z}_{i},\bm{s}_{i}\right) &= \alpha_{n}^{2}\rho(\bm{s}_{i},\bm{z}_{i})\bm{z}_{i}\bm{z}_{i}^{T}K\left(h^{-1}\|\bm{s}-\bm{s}_{i}\|\right)+o\left(h^{2}\right).
\end{align*}
By the symmetry of the kernel function,
\[
E\left(\omega_{i}^{2}\right)=n^{-1}f(\bm{s}){\rm diag}\left\{ \nu_{0},\nu_{2},\nu_{2}\right\} \otimes\bm{\Gamma}(\bm{s})+o\left(1\right).
\]
Thus, 
\[
Var\left(\Omega_{n}\right)=f(\bm{s}){\rm diag}\left\{ \nu_{0},\nu_{2},\nu_{2}\right\} \otimes\bm{\Gamma}(\bm{s})+o\left(1\right).
\]
\end{proof}
\begin{lem}
\label{lemma:delta}
\begin{align*}
E\left[\sum_{i=1}^{n}q_{2}\left(\bm{Z}_{i}^{T}\bm{\zeta}(\bm{s}),Y_{i}\right)\bm{Z}_{i}\bm{Z}_{i}^{T}K_{h}\left(\|\bm{s}-\bm{s}_{i}\|\right)\right] &= -f(\bm{s}){\rm diag}\left\{ \kappa_{0},\kappa_{2},\kappa_{2}\right\} \otimes\bm{\Gamma}(\bm{s})+o\left(1\right)\\
&= -\Delta+o\left(1\right)
\end{align*}
and
\[
Var\left\{ \left(\sum_{i=1}^{n}q_{2}\left(\bm{Z}_{i}^{T}\bm{\zeta}(\bm{s}),Y_{i}\right)\bm{Z}_{i}\bm{Z}_{i}^{T}K_{h}\left(\|\bm{s}-\bm{s}_{i}\|\right)\right)_{ij}\right\} =O\left(n^{-1}h^{-2}\right)
\]
\end{lem}
\begin{proof}
\textbf{Expectation}: The approach is similar to the proof of Lemma
\ref{lemma:omega}. By the Taylor expansion of $q_{2}\left(\bm{z}_{i}^{T}\bm{\zeta}(\bm{s}),\mu\left(\bm{s}_{i},\bm{z}_{i}\right)\right)$
around $\bm{z}_{i}^{T}\bm{\zeta}(\bm{s}_{i})$:
\begin{align*}
q_{2}\left(\bm{z}_{i}^{T}\bm{\zeta}(\bm{s}),\mu(\bm{s}_{i},\bm{z}_{i})\right) &= q_{2}\left(\bm{z}_{i}^{T}\bm{\zeta}(\bm{s}_{i}),\mu(\bm{s}_{i},\bm{z}_{i})\right)+q_{3}\left(\bm{z}_{i}^{T}\bm{\zeta}(\bm{s}_{i}),\mu(\bm{s}_{i},\bm{z}_{i})\right)\left\{ \bm{z}_{i}^{T}\bm{\zeta}(\bm{s})-\bm{z}_{i}^{T}\bm{\zeta}(\bm{s}_{i})\right\} \\
&= -\rho(\bm{s}_{i},\bm{z}_{i})+o\left(1\right).
\end{align*}
And by the same arguments as before
\begin{align*}
E\left(\delta_{i}|\bm{Z}_{i}=\bm{z}_{i},\bm{s}_{i}\right) &= -\alpha_{n}^{2}\rho(\bm{s}_{i},\bm{z}_{i})\bm{z}_{i}\bm{z}_{i}^{T}K\left(h^{-1}\|\bm{s}_{i}-\bm{s}\|\right)\\
E\left(\delta_{i}|\bm{s}_{i}\right) &= -\alpha_{n}^{2}\left(\begin{array}{c}
1\\
h^{-1}(s_{i,1}-s_{1})\\
h^{-1}(s_{i,2}-s_{2})
\end{array}\right)\left(\begin{array}{c}
1\\
h^{-1}(s_{i,1}-s_{1})\\
h^{-1}(s_{i,2}-s_{2})
\end{array}\right)^{T}\otimes\bm{\Gamma}(\bm{s}_{i})K\left(h^{-1}\|\bm{s}_{i}-\bm{s}\|\right)\\
E\left(\delta_{i}\right) &= -nf\left(\bm{s}\right){\rm diag}\left\{ \kappa_{0},\kappa_{2},\kappa_{2}\right\} \otimes\bm{\Gamma}\left(\bm{s}\right)+o\left(n^{-1}\right)
\end{align*}
Thus, 
\[
E\left(\Delta_{n}\right)=-f(\bm{s}){\rm diag}\left\{ \kappa_{0},\kappa_{2},\kappa_{2}\right\} \otimes\bm{\Gamma}(\bm{s})+o\left(1\right)
\]
\textbf{Variance}: From the previous result, it follows that $\left\{ E\left(\delta_{i}\right)\right\} ^{2}=O\left(n^{-2}\right)$.
By the definition of $\delta_{i}$,
\begin{multline*}
E\left(\delta_{i}^{2}|\bm{Z}_{i}=\bm{z}_{i},\bm{s}_{i}\right)=\\
\alpha_{n}^{4}\bm{z}_{i}^{T}\bm{z}_{i}q_{2}^{2}(\bm{s}_{i},\bm{z}_{i})\left(\begin{array}{c}
1\\
h^{-1}(s_{i,1}-s_{1})\\
h^{-1}(s_{i,2}-s_{2})
\end{array}\right)\left(\begin{array}{c}
1\\
h^{-1}(s_{i,1}-s_{1})\\
h^{-1}(s_{i,2}-s_{2})
\end{array}\right)^{T}\bm{z}_{i}\bm{z}_{i}^{T}K^{2}\left(h^{-1}\|\bm{s}_{i}-\bm{s}\|\right)+o\left(1\right)
\end{multline*}
And it follows that $E\left(\delta_{i}^{2}\right)=O\left(n^{-1}\alpha_{n}^{2}\right)$,
and $Var\left(\Delta_{n}\right)=O\left(\alpha_{n}^{2}\right)$.
\end{proof}

\bibliographystyle{chicago}